\documentclass{article}
\usepackage[utf8]{inputenc}
\usepackage[paper=a4paper,left=25mm,right=25mm,top=25mm,bottom=25mm]{geometry}

\usepackage{amsmath, amsfonts, amssymb, mathtools, amsthm}
\usepackage{authblk}
\usepackage{hyperref}
\usepackage{enumitem}
\usepackage{bbm}
\usepackage{tikz}
\usetikzlibrary{bending, positioning, arrows.meta, quotes}
\usepackage{color,soul}
\usepackage{multicol, multirow}
\usepackage{tabularx}

\usepackage{natbib}
\usepackage{bibunits}

\usepackage{ulem}

\defaultbibliographystyle{apalike} 
\defaultbibliography{references}

\newcommand{\PFS}{\text{PFS}}
\newcommand{\OS}{\text{OS}}
\newcommand{\Hglob}{H_{0,\text{global}}}

\newtheorem{lemma}{Lemma}
\newtheorem{theorem}{Theorem}

\usepackage[]{graphicx}

\title{Exhausting the type I error level in event-driven group-sequential designs with a closed testing procedure for progression-free and overall survival}

\author[1,*]{Moritz Fabian Danzer}
\author[2]{Kaspar Rufibach}
\author[3]{Jan Beyersmann}
\author[1]{Rene Schmidt}

\affil[1]{Institute of Biostatistics and Clinical Research, University of Münster, Germany}
\affil[2]{Merck KGaA, Darmstadt, Germany}
\affil[3]{Institute of Statistics, University of Ulm, Germany}
\affil[*]{Corresponding author: moritzfabian.danzer@ukmuenster.de}

\date{\today}

\begin{document}
\begin{bibunit}
    
	\maketitle
	
\begin{abstract}
In oncological clinical trials, overall survival (OS) is the gold-standard endpoint, but long follow-up and treatment switching can delay or dilute detectable effects. Progression-free survival (PFS) often provides earlier evidence and is therefore frequently used together with OS as multiple primary endpoints. Since in certain scenarios trial success may be defined if one of the two hypotheses involved can be rejected, a correction for multiple testing may be deemed necessary. Because PFS and OS are generally highly dependent, their test statistics are typically correlated. Ignoring this dependency (e.g. via a simple Bonferroni correction) is not power optimal.\\
We develop a group-sequential testing procedure for the multiple primary endpoints PFS and OS that fully exhausts the family-wise error rate (FWER) by exploiting their dependence. Specifically, we characterize the joint asymptotic distribution of log-rank statistics across endpoints and multiple event-driven analysis cutoffs. Furthermore, we show that we can consistently estimate the covariance structure. Embedding these results in a closed testing procedure, we can recalculate critical values of the test statistics in order to spend the available type I error optimally. An important extension to the current literature is that we allow for both interim and final analysis to be event-driven.\\
Simulations based on illness–death multi-state models (Markov and non-Markov) empirically confirm FWER control for moderate to large sample sizes. Compared with a simple Bonferroni correction, the proposed methods recover roughly two thirds of the power loss for OS, increase disjunctive and conjunctive power, and enable meaningful early stopping. In planning, these gains translate into about 5\% fewer OS events required to reach the targeted power. We also discuss practical issues in the implementation of such designs and possible extensions of the introduced method.
\end{abstract}

\section{Introduction}\label{sec:introduction}
In oncological clinical trials, the time-to-event endpoint overall survival (OS), which is defined as the time from randomization to death, is the gold standard endpoint because potential clinical benefit can unambiguously be derived from it \citep{Pazdur:2008}. However, as this requires a long follow-up period and the actual effects may be confounded by switches to other treatments after progression, the use of short-term or surrogate endpoints may be informative. The most prominent candidate is progression-free survival (PFS) which is defined as the time from randomization to progression of the disease or death, whatever occurs first. Since the power of the log-rank test depends on the number of events, a hypothesis test for PFS can be performed earlier than for OS. If rejection of either one of the two null hypotheses is sufficient to claim success, this corresponds to ’multiple primary endpoints’ as defined by pertinent guidelines \citep{FDA:2017}. This scenario is relevant in drug development because it may allow for accelerated pathways  from regulatory agencies: accelerated approval by the FDA or conditional approval by the EMA, for example. In addition, availability of an analytical framework also allows to evaluate operating characteristics of futility stopping rules .\\
PFS and OS  are defined as waiting times in an illness-death model. This ensures that PFS is less than or equal to OS without imposing any restrictions on their dependence  \citep{Meller:2019}. A confirmatory analysis of both endpoints may also be worthwhile because it is easy to construct cases in which therapies, that are undoubtedly beneficial, show effects in PFS but not in OS, and vice versa (see \cite{Broglio:2009} and \cite{Morita:2015}, respectively). In \cite{Erdmann:2025}, it is illustrated what hazards can look like to generate such scenarios. In addition, this paper exploits the dependency between PFS and OS to plan a clinical trial such that planning assumptions for both endpoints are consistent. In particular, it becomes apparent that the assumption of proportional hazards can theoretically only hold in rather artificial situations for both endpoints simultaneously. Multiple primary endpoints are accounted for in sample size calculations using a weighted Bonferroni correction.\\
This Bonferroni correction is the simplest way to protect the family-wise error rate (FWER), i.e. the probability of making at least one false discovery, over the two endpoints PFS and OS. Our goal here is to improve on this by properly accounting for the joint distribution of the test statistics. Exploiting dependence between test statistics to gain power has a long tradition in clinical statistics: The prime example are group-sequential trials for one endpoint where the correlation between test statistics over time is considered. Mathematically identical is the scenario of nested populations as first discussed in \cite{Spiessens:2010}. In multi-arm trials with a shared control arm, Dunnett type tests \citep{Dunnett:1955}, which take the dependence of pairwise comparisons into account, are applied.\\
\cite{Wei:1984} investigated the asymptotic joint distribution of log-rank test statistics for potentially dependent time-to-event endpoints and \cite{Lin:1991} extended this to group-sequential designs. Beyond that, a challenge in our scenario is that the analysis cutoffs are chosen on an event-driven basis, i.e. after a certain number of events of a particular type has been observed. Event-driven censoring leads to both dependence of time-to-event and time-to-censoring and to dependence of the observed data across units. \cite{Rühl:2023} demonstrated recently that this type of censoring is generally compatible with the common assumptions of analyses of time-to-event endpoints in a counting process approach. However, these authors also found that including calendar time information in an event-driven analysis may introduce bias, as it disturbs the intensity of a counting process. Hence, in our setting, the case may be somewhat more complicated, as we are also interested in examining OS at the time of an analysis triggered by PFS events and vice versa. In this context, we take advantage of the fact that the stochastic process of log-rank statistics asymptotically behaves like a time-transformed Brownian motion in calendar time \citep{Olschewski:1986}. The line of argument will be as follows: Assuming independence of the uncensored patient data, we allow for both staggered trial entry and an event-driven final analysis based on the recent result of \cite{Rühl:2023} and using that PFS and OS arise from counting processes of the illness-death model. This approach does not require asymptotic arguments. To ensure that the intensity of the counting processes at hand are not disturbed by an event-driven interim, we demonstrate that the latter time point can asymptotically be replaced by a deterministic time which is determined via the trial design.\\
Having provided a framework, where both interim and final analysis may be event-driven, justifying current practice, we also aim at an improvement over the typically employed simple Bonferroni correction. This can be achieved by using closed testing procedures as introduced by \cite{Marcus:1976}. In applications, the graphical procedures of \cite{Bretz:2011, Maurer:2013} in particular have proven to be extremely helpful. The recently published work \cite{Anderson:2022} showed how these closed testing procedures can be further improved in terms of power by explicitly exploiting the known or consistently estimable correlation structure of test statistics. This applies in particular to group-sequential designs, in which the correlation across the various endpoints and analysis times must be taken into account. Compared to previous approaches, which, for example, define a hierarchical order of endpoints \citep{Glimm:2010}, this framework offers flexibility, which we consider to be advantageous for the reasons mentioned above.\\
The paper is organised as follows. In Section \ref{sec:asymptotic_distribution}, we introduce notation, present test statistics and their asymptotic joint distribution. We show how we can apply this within the framework of \cite{Anderson:2022} in a testing procedure for the two endpoints PFS and OS based on two exemplary designs in Section \ref{sec:sequential_designs}. Section \ref{sec:simulation_studies} contains the results of our simulation studies. In Section \ref{sec:practical_issues}, we address some practical issues connected to the implementation of the presented design. We conclude with a discussion in Section \ref{sec:discussion}. Proofs and additional simulation results are in the Supplementary Material. The code underlying our simulation study and the complete results are available at \url{https://github.com/moedancer/MultSurvTrialDesign}.

\section{Joint distribution of PFS- and OS-test statistics}\label{sec:asymptotic_distribution}

Each patient is recruited at calendar time $R$, assigned to a treatment group $Z \in \{0,1\}$, and experiences events PFS and OS at random times $T_{\PFS}$ and $T_{\OS}$ after its recruitment. Event dates might be randomly censored due to drop-out at time $C$ after enrollment. It is important to distinguish censoring through $C$ from administrative censoring. Administrative censoring is typically event-driven, i.e. done after a certain number of events has been observed in the trial. The observable data at calendar time $t$ for a patient who has already been recruited thus reduces to the tuple $(Z, R, X_{\PFS}(t), \Delta_{\PFS}(t), X_{\OS}(t), \Delta_{\OS}(t))$ where 
\begin{align*}
	&X_{\PFS}(t) \coloneqq T_{\PFS} \wedge C \wedge (t - R)_+, \Delta_{\PFS}(t)\coloneqq \mathbbm{1}_{T_{\PFS} \leq C \wedge (t - R)_+} \text{ and }\\
	&X_{\OS}(t) \coloneqq T_{\OS} \wedge C \wedge (t - R)_+, \Delta_{\OS}(t)\coloneqq \mathbbm{1}_{T_{\OS} \leq C \wedge (t - R)_+}
\end{align*}
for each patient where $a \wedge b$ denotes the minimum of the real numbers $a,b$. In particular, this information can be used to reconstruct when the individual was at risk for progression or death during the course of the trial. In a clinical trial we have $n$ independent replicates of this tuple at time $t$. These will be indexed by $i \in \{1,\dots,n\}$. The planned number of patients $n$ is fixed.\\
In an event-driven design, analyses will be conducted when a proportion of at least $r_{\PFS} \in [0,1)$ resp. $r_{\OS} \in [0,1)$ of the $n$ patients have experienced a PFS or an OS event, respectively.\\
The PFS analysis will be conducted at the random analysis cutoff date
\begin{equation}\label{eq:caldate_analysis_pfs}
	A_{\PFS} \coloneqq \inf \left\{ t \geq 0 \colon \frac{1}{n} \sum_{i=1}^n \Delta_{\PFS,i}(t) \geq r_{\PFS} \right\}
\end{equation} 
and the OS takes place at the random analysis cutoff date
\begin{equation}\label{eq:caldate_analysis_os}
	A_{\OS} \coloneqq \inf \left\{ t \geq 0 \colon \frac{1}{n} \sum_{i=1}^n \Delta_{\OS,i}(t) \geq r_{\OS} \right\}.
\end{equation}
In other words, we perform an interim or the final analysis as soon as the targeted number of events $d_E \coloneqq \lceil r_E \cdot n \rceil$ for the respective endpoint $E \in \{\PFS,\OS\}$ has been observed. In addition, we want to have the flexibility to perform an analysis for OS at the time of PFS analysis, and vice versa. The chronological order $A_{\PFS} \leq A_{\OS}$ is guaranteed if $r_{\PFS} \leq r_{\OS}$. According to the definitions in \eqref{eq:caldate_analysis_pfs} and \eqref{eq:caldate_analysis_os}, the analysis time might be equal to $\infty$ if too many patients are lost to follow-up. Suitable measures must be taken to prevent this, e.g. by choosing $r_{\PFS}$ and $r_{\OS}$ carefully,  by choosing a maximal calendar time for the respective analyses in advance, and operational measures to prevent excessive drop-out. If event-driven analysis cutoffs are determined in this way, these analysis cutoffs converge in probability to
\begin{equation}\label{eq:limits_analysis_dates}
	\begin{split}
	t_{\PFS}\coloneqq \inf \left\{ t\geq 0 \colon \mathbb{P}[T_{\PFS} \leq C \wedge (t - R)_+ ] \geq r_{\PFS} \right\} & \; \text{and}\\
	t_{\OS}\coloneqq \inf \left\{ t\geq 0 \colon \mathbb{P}[T_{\OS} \leq C \wedge (t - R)_+ ] \geq r_{\OS} \right\}
	\end{split}
\end{equation}
as $n \to \infty$. This is stated in Lemma \ref{lemma:calendar_dates}.\\
Next, we present notation and test statistics previously introduced in \cite{Lin:1991}. For each patient, the counting process $(N_{E,i}(t,s))_{t,s\geq 0}$ denotes whether the event $E \in \{\PFS,\OS\}$ has been observed at calendar time $t$ and before the patient has spent time $s$ in the trial, i.e.
\begin{equation*}
	N_{E,i}(t,s) \coloneqq \Delta_{E,i}(t) \cdot \mathbbm{1}_{T_{E,i} \leq s}.
\end{equation*}
Analogously, we define the at risk processes $(Y_{E,i}(t,s))_{t,s \geq 0}$. It indicates whether the patient is still at risk of experiencing event $E \in \{\PFS,\OS\}$ after if already spent time $s$ in the trial. However, we only consider information that is available up to calendar time $t$. In particular, this implies that $Y_{E,i}(t,s) = 0$ whenever $s \geq (t-R_i)_+$. It is given by
\begin{equation*}
	Y_{E,i}(t,s) \coloneqq \mathbbm{1}_{X_{E,i}(t) \geq s}.
\end{equation*}
Both quantities can be aggregated over the entire population. Those aggregates are denoted by $N_E(t,s)\coloneqq \sum_{i=1}^n N_{E,i}(t,s)$ and $Y_E(t,s)\coloneqq \sum_{i=1}^n Y_{E_i}(t,s)$, respectively. The first of these processes denotes the number of events of type $E$ that were observed until calendar time $t$ and that happened before the respective patients have spent time $s$ in the trial. The second one is the number of patients that have spent time $s$ in the trial without being censored or experiencing the event $E$ up to calendar time $t$. Obviously, we have $N_E(t,s) = N_E(t,t)$ and $Y_E(t,s) = 0$ for $s \geq t$. For the at risk-processes we also consider the group-specific quantities aggregate for the group with $Z=g$ with $g \in \{0,1\}$ by
\begin{equation*}
	Y^{Z=g}_E(t,s)\coloneqq \sum_{i=1}^n \mathbbm{1}_{Z_i=g} \cdot Y_{E,i}(t,s)
\end{equation*}
for $E \in \{\PFS, \OS\}$.\\
At calendar time $t$ the log-rank test statistic for the time-to-event endpoint $E \in \{\PFS,\OS\}$ is given by
\begin{equation*}
	U_E(t) \coloneqq \frac{1}{\sqrt{n}} \sum_{i=1}^n \int_0^t \left(Z_i - \frac{Y^{Z=1}_{E}(t,s)}{Y_{E}(t,s)}\right) N_{E,i}(t,ds) = \frac{1}{\sqrt{n}} \sum_{i=1}^n \Delta_{E,i}(t) \left( Z_i - \frac{Y^{Z=1}_{E}(t,X_{E,i}(t))}{Y_{E}(t,X_{E,i}(t))} \right).
\end{equation*}
The expected value of the processes $Y_E$ and $Y_E^{Z=1}$ are given by
\begin{equation*}
	\frac{1}{n}\mathbb{E}\left[Y_{E}(t,s)\right]=y_{E}(t,s) \coloneqq \mathbb{P}[(T_E \wedge C) > \max(s,(t - R)_+); R \leq t]
\end{equation*}
and
\begin{equation*}
	\frac{1}{n}\mathbb{E}\left[Y^{Z=1}_{E}(t,s)\right]=y^{Z=1}_{E}(t,s) \coloneqq \mathbb{P}[(T_E \wedge C) > \max(s,(t - R)_+); R \leq t; Z=1],
\end{equation*}
respectively, where, $\mathbb{P}[A;B]$ is the probability of the intersection of the events $A$ and $B$. These expected values denote the probability that a random patient has spent at least time $s$ in the trial without experiencing event $E$ and without being censored up to calendar date $t$. For each fixed $t\geq 0$, we now consider the process $(N_{E,i}(t,s))_{s \geq 0}$. By $\mathcal{F}(t,s)$ we denote the $\sigma$-algebra that contains the information about all events that happen before calendar time $t$ and within the calendar time interval $[R_i, R_i + s)$ for each patient $i$. For each fixed $t$, the corresponding counting process martingale w.r.t. the filtration $(\mathcal{F}(t,s))_{s \geq 0}$ is given by
\begin{equation*}
	M_{E,i}(t,s) = N_{E,i}(t,s) - \int_0^s Y_{E,i}(t,u) \lambda_E(u) du
\end{equation*}
where $\lambda_E$ denotes the hazard function of experiencing event $E$. It is given by
\begin{equation*}
	\lambda_E(s)\coloneqq \lim_{h \searrow 0} \frac{\mathbb{P}[T_E \in [t,t+h)|T_E \geq t]}{h} = \frac{f_{T_E}(s)}{S_{T_E}(s)}
\end{equation*}
where $f$ and $S$ shall denote probability density and survival function of the indexed random variable. Under the null hypothesis of equal distributions of the event $E$ in both treatment groups, the process $(U_E(t))_{t \geq 0}$ is asymptotically equivalent to the process $(u_E(t))_{t \geq 0}$. This latter process is defined by
\begin{equation*}
	u_E(t) \coloneqq \frac{1}{\sqrt{n}} \sum_{i=1}^n \int_0^t \left( Z_i - \frac{y^{Z=1}_{E}(t,s)}{y_{E}(t,s)} \right) M_{E,i}(t,ds).
\end{equation*}
Here, we replaced the counting process $N$ by the corresponding martingale $M$ in the integrator and the aggregated at risk processes $Y$ by their expectations $y$. Among others, \cite{Tsiatis:1981, Sellke:1983, Lin:1991} demonstrated the validity of these replacements. Going beyond that, we want to consider event-driven analyses of (possibly) both endpoints at the random analysis cutoffs $A_{\PFS}$ and $A_{\OS}$. These random cutoffs are called when a share of $r_{\PFS}$ resp. $r_{\OS}$ of the total $n$ patients have reached their PFS resp. OS event. For these analyses Theorem \ref{thm:main_convergence} yields
\begin{equation*}
	\mathbf{U}_{\PFS, \OS} \coloneqq (U_{\PFS}(A_{\PFS}), U_{\OS}(A_{\PFS}), U_{\PFS}(A_{\OS}), U_{\OS}(A_{\OS})) \overset{\mathcal{D}}{\to} \mathcal{N}(0, \boldsymbol{\Sigma}_{\PFS, \OS}).
\end{equation*}
In the limit $n \to \infty$, $A_{\PFS}$ and $A_{\OS}$ will converge against the fixed dates $t_{\PFS}$ and $t_{\OS}$ (see Lemma \ref{lemma:calendar_dates}). The asymptotic covariance matrix $\boldsymbol{\Sigma}_{\PFS, \OS}$ is then given by
\begin{equation*}
	\boldsymbol{\Sigma}_{\PFS, \OS} = 
	\begin{pmatrix}
		\sigma^2_{\PFS}(t_{\PFS}) & \sigma_{\PFS, \OS}(t_{\PFS}, t_{\PFS}) & \sigma^2_{\PFS}(t_{\PFS}) & \sigma_{\PFS, \OS}(t_{\PFS}, t_{\OS})\\
		\sigma_{\PFS, \OS}(t_{\PFS}, t_{\PFS}) & \sigma^2_{\OS}(t_{\PFS}) & \sigma_{\PFS, \OS}(t_{\OS}, t_{\PFS}) & \sigma^2_{\OS}(t_{\PFS})\\
		\sigma^2_{\PFS}(t_{\PFS}) & \sigma_{\PFS, \OS}(t_{\OS}, t_{\PFS}) & \sigma^2_{\PFS}(t_{\OS}) & \sigma_{\PFS, \OS}(t_{\OS}, t_{\OS})\\
		\sigma_{\PFS, \OS}(t_{\PFS}, t_{\OS}) & \sigma^2_{\OS}(t_{\PFS}) & \sigma_{\PFS, \OS}(t_{\OS}, t_{\OS}) & \sigma^2_{\OS}(t_{\OS})
	\end{pmatrix}
\end{equation*}
For a concise description of the estimation of the components of the matrix, we have to introduce
\begin{equation*}
	\hat{\mu}^{Z=g}_E(t,s)\coloneqq 1 - \frac{Y_E^{Z=g}(t,s)}{Y_E(t,s)} \quad \text{and} \quad \hat{\psi}^{Z=g}_{E}(t,s) \coloneqq \int_0^s \hat{\mu}^{Z=g}_E(t,u) \hat{\Lambda}_E(t,du)
\end{equation*}
for $E \in \{\PFS, \OS\}$ and $g\in\{0,1\}$ where $\hat{\Lambda}_E(t,\cdot)$ denotes the Nelson-Aalen estimate of the cumulative hazard function for event $E$ from all data available at calendar time $t$.\\
For components of the covariance matrix that refer to the same endpoint, i.e. those of the form $\sigma^2_{E_1}(t_{E_2})$ we can use standard estimates for log-rank test statistics, evaluated at the random analysis cutoff date $A_{E_2}$, i.e.
\begin{align*}
	\hat{\sigma}^2_{E_1}(A_{E_2}) \coloneqq &\frac{1}{n} \sum_{i=1}^n \int_0^{A_{E_2}} \frac{Y_{E_1}^{Z=1}(A_{E_2},s)}{Y_{E_1}(A_{E_2},s)}  \left( 1 - \frac{Y_{E_1}^{Z=1}(A_{E_2},s)}{Y_{E_1}(A_{E_2},s)}   \right) N_{E_1,i}(A_{E_2},ds)\\
	=&\frac{1}{n} \sum_{i=1}^n \int_0^{A_{E_2}} \hat{\mu}^{Z=0}_{E_1}(A_{E_2},s) \cdot  \hat{\mu}^{Z=1}_{E_1}(A_{E_2},s) \; N_{E_1,i}(A_{E_2},ds)\\
	=&\frac{1}{n} \sum_{i=1}^n \hat{\mu}^{Z=0}_{E_1}(A_{E_2}, X_{E_1,i}(A_{E_2})) \cdot  \hat{\mu}^{Z=1}_{E_1}(A_{E_2}, X_{E_1,i}(A_{E_2})) \cdot \Delta_{E_1,i}(A_{E_2})
\end{align*}
for $E_1, E_2 \in \{\PFS, \OS\}$. For components addressing the covariance of test statistics for different endpoints, i.e. those of the form $\sigma^2_{\PFS,\OS}(t_{E_1}, t_{E_2})$ the estimate amounts to
\begin{align*}
	&\hat{\sigma}_{\PFS, \OS}(A_{E_1}, A_{E_2}) \\
	\coloneqq & \frac{1}{n} \sum_{i=1}^{n} \Bigg( \left( \hat{\mu}^{Z=Z_i}_{\PFS}(A_{E_1}, X_{\PFS,i}(A_{E_1})) \Delta_{\PFS,i}(A_{E_1}) - \hat{\psi}^{Z=Z_i}_{\PFS}(A_{E_1}, X_{\PFS,i}(A_{E_1})) \right) \cdot \\
	&\qquad \qquad \left( \hat{\mu}^{Z=Z_i}_{\OS}(A_{E_2}, X_{\OS,i}(A_{E_2})) \Delta_{\OS,i}(A_{E_2}) - \hat{\psi}^{Z=Z_i}_{\OS}(A_{E_2}, X_{\OS,i}(A_{E_2})) \right) \Bigg)\\
\end{align*}
for $E_1, E_2 \in \{\PFS, \OS\}$. As shown in Theorem \ref{theorem:consistency_variance}, this constitutes a consistent variance estimation mechanism, i.e.
\begin{equation*}
	\hat{\boldsymbol{\Sigma}}_{\PFS, \OS} \overset{\mathbb{P}}{\to} \boldsymbol{\Sigma}_{\PFS, \OS}
\end{equation*}
where our estimate $\hat{\boldsymbol{\Sigma}}_{\PFS, \OS}$ contains all the components mentioned above. These are evaluated at the random analysis cutoffs $A_{\PFS}$ and $A_{\OS}$. Proofs of all statements are deferred to the Supplementary Material. They combine results derived in \cite{Wei:1984, Lin:1991} with Empirical Process Theory as presented in \cite{Shorack:2009} to account for event-driven censoring.\\

\section{Exhausting the type I error rate in a group-sequential procedure}\label{sec:sequential_designs}

The previous observations open up the possibility of exploiting this dependency in a group-sequential test procedure as described e.g. in \cite{Anderson:2022}. As in \cite{Lin:1991}, we are interested in testing the two null hypotheses
\begin{equation*}
	H_{0, \PFS}\colon S^{Z=0}_{\PFS}(s) = S^{Z=1}_{\PFS}(s) \quad \forall s\geq 0 \quad \text{and} \quad H_{0, \OS}\colon S^{Z=0}_{\OS}(s) = S^{Z=1}_{\OS}(s) \quad \forall s\geq 0.
\end{equation*}
Here, $S^{Z=g}_{E}$ denotes the survival function of endpoint $E$ in the respective groups $g \in \{0,1\}$. As we are interested in detecting a potential superiority of the experimental treatment with respect to at least one of the two hypotheses, we will apply one-sided test. We want to test those hypotheses within a closed testing procedure as described by \cite{Marcus:1976}. In order to reject any of the endpoint-specific hypotheses $H_{0, \PFS}$ and $H_{0, \OS}$ we first have to reject the intersection null hypothesis 
\begin{equation*}
	\Hglob \coloneqq H_{0, \PFS} \cap H_{0, \OS}.
\end{equation*}
As suggested by \cite{Anderson:2022}, we follow a weighting strategy from the graphical approach introduced in \cite{Bretz:2011} to examine $\Hglob$ and its components. In particular, we split up our overall significance level $\alpha$ into $\rho_{\PFS}\alpha$ and $\rho_{\OS}\alpha$ with $0 < \rho_{\PFS}, \rho_{\OS} < 1$ and $\rho_{\PFS} + \rho_{\OS} = 1$ which shall be used for the respective hypotheses. In case of a rejection of one of $H_{0,\PFS}$ and $H_{0,\OS}$ the level shall be propagated to the remaining component. This is depicted by the weighted directed graph in Figure \ref{figure:weighting_graph}.\\
\begin{figure}
	\centering
	\begin{tikzpicture}
		
		\node [circle, draw, node distance=2cm] (pfs) 
		{\parbox{1.5cm}{\centering $H_{0,\PFS}$ \\$\rho_{\PFS}$} };
		\node [circle, draw, node distance=2cm, right = of pfs] (os) 
		{\parbox{1.5cm}{\centering $H_{0,\OS}$ \\$\rho_{\OS}$} };
		
		\draw (pfs) edge ["1", ->, bend left = 45] (os);
		\draw (os) edge ["1", ->, bend left = 45] (pfs);
		
	\end{tikzpicture}
	\caption{Graphical representation for weighting strategy of the multiple testing approach with $\rho_{\PFS} + \rho_{\OS} = 1$.}
	\label{figure:weighting_graph}
\end{figure}
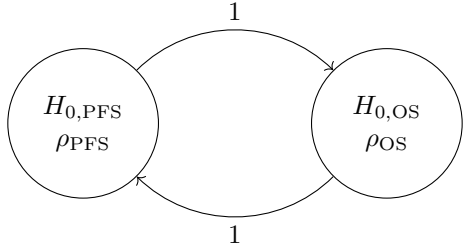%
Furthermore, we have to specify endpoint-specific $\alpha$-spending function families $g_{\PFS}\colon [0,1] \times [0,\alpha] \to [0, \alpha]$ and $g_{\OS}\colon [0,1] \times [0, \alpha] \to [0, \alpha]$. Both need to be monotonically increasing in their first argument. For time-to-event endpoints, this argument is typically given as the information fraction for the respective endpoint which for a time-to-event endpoint is the proportion of events observed so far in relation to the targeted number of events $d_E$. Hence, at some calendar time $t$, the current information fraction for the endpoint $E$ is given by
\begin{equation*}
	\tau_E(t) \coloneqq \frac{\sum_{i=1}^n N_E(t,t)}{d_E}.
\end{equation*}
The second argument is given by the total type I error level that shall be spent on this endpoint. To guarantee this, we require $g_E(\tau ,\rho \alpha) \leq \rho \alpha$ for all $0 \leq \rho,\tau \leq 1$. In the setting of the graphical testing procedure of Figure \ref{figure:weighting_graph}, we shall only spend $\rho_{\PFS}\alpha$ on PFS and $\rho_{\OS}\alpha$ on OS as long as the joint null hypothesis has not been rejected. Hence, we need to specify $g_{\PFS}(\cdot, \rho_{\PFS} \alpha)$ and $g_{\OS}(\cdot, \rho_{\OS} \alpha)$. However, as soon as the joint null hypothesis $\Hglob$ is rejected based on one of the endpoints, the other endpoint can be tested at full level $\alpha$ according to the graphical procedure. This is why we must also specify $g_{\PFS}(\cdot,\alpha)$ and $g_{\OS}(\cdot,\alpha)$. We would like to emphasize that these functions must  be pre-specified at the trial design stage.\\
As soon as the targeted number of events is reached, we want to have exhausted the significance level. After that, no further significance level should be spent for the endpoint. Furthermore, no significance level should be spent before the first event has been observed. In order to meet these requirements, we also assume that
\begin{equation*}
	g_E(0,\rho\alpha) = 0 \quad \text{and} \quad g_E(1,\rho\alpha) = \rho\alpha \quad \forall 0 \leq \rho \leq 1.
\end{equation*}
Alpha-spending functions are further discussed in Section 7 of \cite{Jennison:2000} or Section 3.3 of \cite{Wassmer:2016}.\\
In the following two subsections, we explain how we can improve the sequential testing procedures with the co-primary endpoints PFS and OS presented \cite{Erdmann:2025} in terms of power while still maintaining the family-wise error rate. We will do this using the tools mentioned here and based on the asymptotical results presented in Section \ref{sec:asymptotic_distribution}. We will assume throughout that $A_{\PFS} \leq A_{\OS}$. This will obviously be the case if $r_{\PFS} \leq r_{\OS}$.
\subsection{No $\alpha$-spending for OS in the first analysis}\label{subsec:no_early_os}
As described above, we are mainly interested in the investigation of PFS at calendar time $A_{\PFS}$ and of OS at calendar time $A_{\OS}$, respectively. We therefore first deal with the case in which no inference about the null hypothesis for OS is planned at the first analysis at calendar time $A_{\PFS}$. The corresponding testing strategy is illustrated in Figure \ref{figure:flow_graph_I}. 
\begin{figure}
	\centering
	\includegraphics[width=.8\textwidth]{"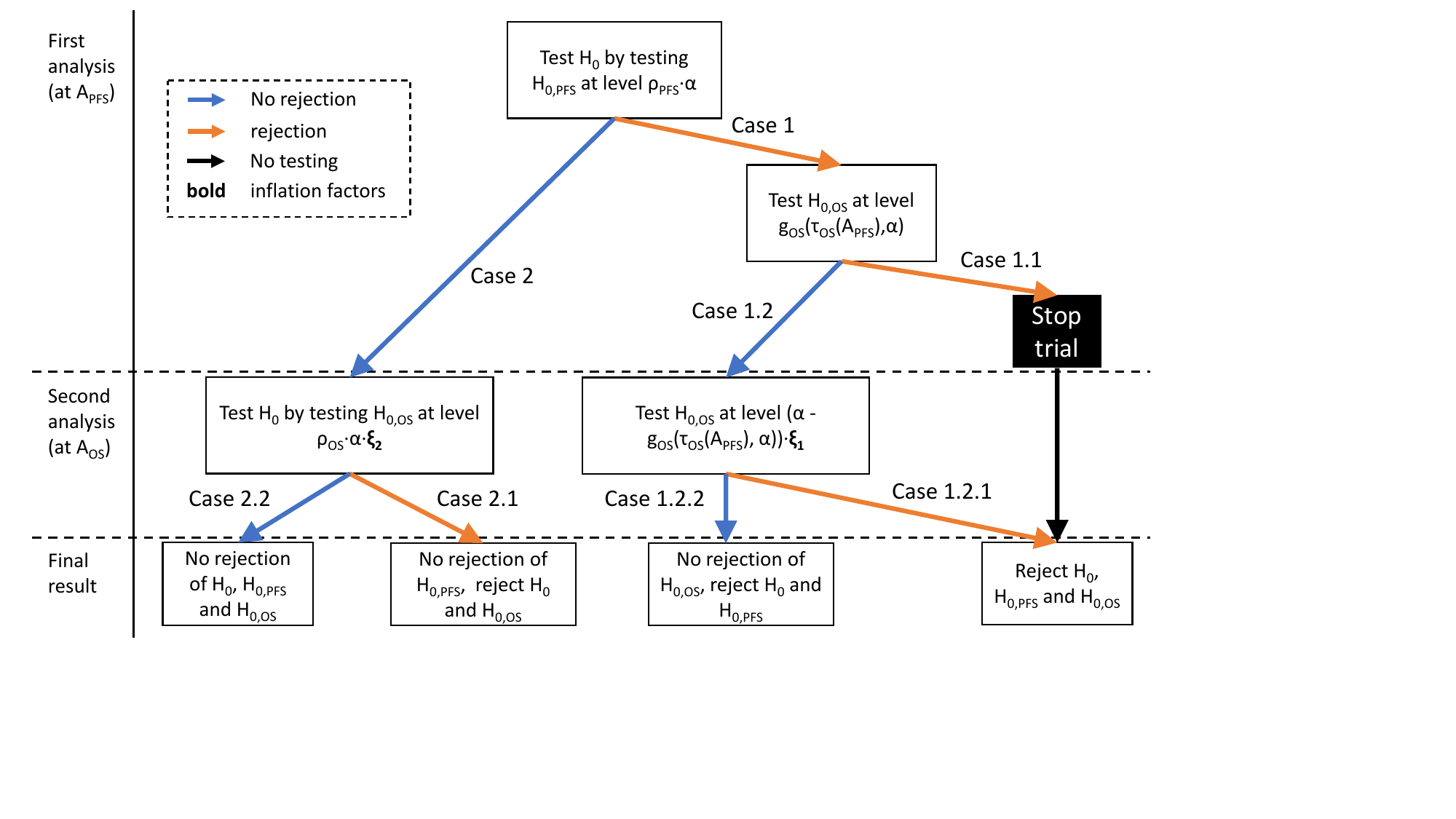"}
	\caption{Flow chart showing the course of a trial without initial testing of $H_{0,\OS}$ at $A_{\PFS}$.}
	\label{figure:flow_graph_I}
\end{figure}
As long as the intersection null hypothesis has not been rejected, we spend the available level for PFS ($\rho_{\PFS}\alpha$) at the first analysis and the available level for OS ($\rho_{\PFS}\alpha$) at the second analysis. This is expressed by the two alpha-spending functions
\begin{equation}\label{eq:design_I_initial_spending}
	g_{\PFS}(s,\rho_{\PFS}\alpha) = \mathbbm{1}_{s \geq 1} \rho_{\PFS}\alpha \quad \text{and} \quad g_{\OS}(s,\rho_{\OS}\alpha) = \mathbbm{1}_{s \geq 1} \rho_{\OS}\alpha.
\end{equation}
In particular, we have $g_{\OS}(\tau_{\OS}(A_{\PFS}), \rho_{\OS}\alpha) = 0$. Please also note that the approach presented in \cite{Anderson:2022} is robust to analysis schedules based on different information times as is intended with these choices.
\paragraph{First analysis}
The first analysis occurs at $A_{\PFS}$ and initially proceeds as in \cite{Erdmann:2025}. As determined by the $\alpha$-spending functions in \eqref{eq:design_I_initial_spending}, we investigate $\Hglob$ by testing $H_{0,\PFS}$ at level $g_{\PFS}(\tau_{\PFS}(A_{\PFS}),\rho_{\PFS}\alpha) = g_{\PFS}(1,\rho_{\PFS}\alpha) = \rho_{\PFS}\alpha$. Our test decision is thus given by
\begin{equation}\label{eq:design_I_PFS}
	\frac{U_{\PFS}(A_{\PFS})}{\sqrt{\hat{\sigma}_{\PFS}^2 (A_{\PFS})}} \leq \Phi^{-1}(\rho_{\PFS}\alpha)
\end{equation}
where $\Phi^{-1}$ denotes the quantile function of the standard normal distribution. We can distinguish between two cases. Either \eqref{eq:design_I_PFS} is met (Case 1) or not (Case 2). For the latter case, we proceed to the second stage without propagating any level as no test of $H_{0,\OS}$ is planned according to the $\alpha$-spending function.\\

\underline{Case 1:} We can reject the joint null hypothesis $\Hglob$. From a formal point of view, we still have to investigate whether we can reject $H_{0,\PFS}$ in order to comply with the framework of \cite{Anderson:2022}. However, this should be the case for a sensibly chosen $\alpha$-spending function $g_{\PFS}(\cdot, \alpha)$ as one should obviously choose $g_{\PFS}(\cdot, \alpha) \geq g_{\PFS}(\cdot, \rho_{\PFS}\alpha)$. Now, we can propagate the level $\rho_{\PFS}\alpha$ to the test of the remaining hypothesis $H_{0,\OS}$. The testing strategy now depends on $g_{\OS}(\cdot, \rho_{\PFS}\alpha + \rho_{\OS}\alpha)=g_{\OS}(\cdot, \alpha)$. At the interim analysis the test decision is determined by
\begin{equation}\label{eq:design_I_pureOS_early}
	\frac{U_{\OS}(A_{\PFS})}{\sqrt{\hat{\sigma}_{\OS}^2 (A_{\PFS})}} \leq \Phi^{-1}(g_{\OS}(\tau_{\OS}(A_{\PFS}), \alpha)).
\end{equation}
If \eqref{eq:design_I_pureOS_early} is fulfilled, we can terminate the trial as we can claim success in rejecting $H_{0,\OS}$ (Case 1.1). Otherwise, $H_{0,\OS}$ remains unrejected for now and we proceed to the next stage (Case 1.2). This may be due to the fact that the evidence for rejecting $H_{0,\OS}$ is not yet convincing or also because $g_{\OS}(\cdot, \alpha)$ does not plan to test $H_{0,\OS}$ at this stage, i.e. $g_{\OS}(\tau_{\OS}(A_{\PFS}), \alpha)=0$. In this context, the question arises as to how we deal with the propagated level. We could spend it immediately or save it for the final analysis. These choices are represented by the two $\alpha$-spending functions
\begin{align}
	&g_{\OS}(s,\alpha) = \mathbbm{1}_{s \geq \tau_{\OS}(A_{\PFS})} \rho_{\PFS}\alpha + \mathbbm{1}_{s \geq 1} \rho_{\OS}\alpha \label{eq:design_I_full_spending_early} \\
	\text{and} \quad &g_{\OS}(s,\alpha) = \mathbbm{1}_{s \geq 1} \alpha, \label{eq:design_I_full_spending_late}
\end{align} 
respectively.\\
\paragraph{Second analysis}
The analysis takes place as soon as the targeted number of OS events has been observed, i.e. at the random analysis cutoff date $A_{\OS}$. As lined out above, we will carry out analyses at this analysis date in the cases 1.2 and 2 at which we will look separately now.\\

\underline{Case 1.2:} We basically proceed with testing $H_{0,\OS}$ as in a group-sequential design that is determined by the $\alpha$-spending function $g_{\OS}(\cdot, \alpha)$. At this analysis we can spend the remaining level $\alpha - g_{\OS}(\tau_{\OS}(A_{\PFS}), \alpha)$.As usual, we can inflate this level by some factor, say $\xi_1$ to ensure
\begin{equation}\label{eq:design_I_pureOS_late}
	\begin{split}
	&\mathbb{P}_{H_{0,\OS}}\left[\frac{U_{\OS}(A_{\PFS})}{\sqrt{\hat{\sigma}_{\OS}^2 (A_{\PFS})}} > \Phi^{-1}(g_{\OS}(\tau_{\OS}(A_{\PFS}), \alpha)) \bigcap \frac{U_{\OS}(A_{\OS})}{\sqrt{\hat{\sigma}_{\OS}^2 (A_{\OS})}} \leq \Phi^{-1}((\alpha - g_{\OS}(\tau_{\OS}(A_{\PFS}), \alpha)) \cdot \xi_1) \right]\\
	= & \alpha - g_{\OS}(\tau_{\OS}(A_{\PFS}), \alpha).
	\end{split}
\end{equation}
As the two test statistics are jointly normally distributed and, according to Theorem \ref{theorem:consistency_variance}, we can consistently estimate the covariance matrix, $\xi_1$ can be computed easily. This corresponds to the standard procedure of group-sequential designs for one endpoint. If the OS test statistic is significant at this inflated level, we can also reject $H_{0,\OS}$ (Case 1.2.1) or remain only with rejection of $H_{0}$ and $H_{0, \PFS}$ (Case 1.2.2). However, one should note that likely, further PFS events will have happened. As discussed in \cite{Asikanius:2024}, this 'pipeline data' is not used for decision-making anymore, but may be used to update estimates of group-specific estimates of survival functions and relative effect measures.\\

\underline{Case 2:} We still want to reject $\Hglob$. In the first analysis, we already spent a level of $\rho_{\PFS}\alpha$ on testing it based on PFS data. In this stage, we intend to spend the remaining $\rho_{\OS}\alpha$ on testing it based on OS data. However, if we use this as the local level, we obtain
\begin{equation}\label{eq:design_I_inefficiency}
	\begin{split}
	&\mathbb{P}_{\Hglob}[\Hglob\text{ can be rejected}]\\
	=&\mathbb{P}_{\Hglob}\left[\frac{U_{\PFS}(A_{\PFS})}{\sqrt{\hat{\sigma}_{\PFS}^2 (A_{\PFS})}} \leq \Phi^{-1}(\rho_{\PFS}\alpha) \bigcup \frac{U_{\OS}(A_{\OS})}{\sqrt{\hat{\sigma}_{\OS}^2 (A_{\OS})}} \leq \Phi^{-1}(\rho_{\OS}\alpha)\right]\\
	\leq&\mathbb{P}_{\Hglob}\left[\frac{U_{\PFS}(A_{\PFS})}{\sqrt{\hat{\sigma}_{\PFS}^2 (A_{\PFS})}} \leq \Phi^{-1}(\rho_{\PFS}\alpha)\right] + \mathbb{P}_{\Hglob}\left[\frac{U_{\OS}(A_{\OS})}{\sqrt{\hat{\sigma}_{\OS}^2 (A_{\OS})}} \leq \Phi^{-1}(\rho_{\OS}\alpha)\right]\\
	=&\rho_{\PFS}\alpha + \rho_{\OS}\alpha= \alpha.
	\end{split}
\end{equation}
The discrepancy between the left and the right hand side of this inequality grows with increasing correlation of the two involved test statistics. In a standard group-sequential design, one can overcome this inefficiency as demonstrated in \eqref{eq:design_I_pureOS_late} as the correlation structure of the test statistics is known. However, based on our results summarized in Section \ref{sec:asymptotic_distribution}, this can also be done here. As in \cite{Anderson:2022}, we can make sure to really spend the full level by calculating the inflation factor $\xi_2$ that fulfills
\begin{equation*}
	1 - \mathbb{P}_{\Hglob}\left[ \frac{U_{\PFS}(A_{\PFS})}{\sqrt{\hat{\sigma}_{\PFS}^2 (A_{\PFS})}} > \Phi^{-1}(\rho_{\PFS}\alpha) \;\bigcap\; \frac{U_{\OS}(A_{\OS})}{\sqrt{\hat{\sigma}_{\OS}^2 (A_{\OS})}} > \Phi^{-1}(\xi_2 \cdot \rho_{\OS}\alpha) \right] = \alpha
\end{equation*}
and determining the rejection of $\Hglob$ by
\begin{equation}\label{eq:design_I_joint_decision_OS}
	\frac{U_{\OS}(A_{\OS})}{\sqrt{\hat{\sigma}_{\OS}^2 (A_{\OS})}} \leq \Phi^{-1}(\xi_2 \cdot \rho_{\OS}\alpha).
\end{equation}
We can calculate $\xi$ because under the strict null hypothesis of equal distribution of all involved endpoints in both groups the two test statistics are centered and asymptotically jointly normally distributed with a covariance matrix that is consistently estimated by 
\begin{equation*}
	\begin{pmatrix}
		\hat{\sigma}^2_{\PFS}(t_{\PFS}) & \hat{\sigma}_{\PFS, \OS}(t_{\PFS}, t_{\OS})\\
		\hat{\sigma}_{\PFS, \OS}(t_{\PFS}, t_{\OS}) & \hat{\sigma}^2_{\OS}(t_{\OS}).
	\end{pmatrix}
\end{equation*} 
If \eqref{eq:design_I_joint_decision_OS} holds, we reject $\Hglob$. We can also reject $H_{0, \OS}$ within the closed testing procedure if we assume that $g_{\OS}(1, \alpha) - g_{\OS}(\tau_{\OS}(A_{\PFS}), \alpha) \geq \rho_{\OS}\alpha$. Now, we could also reinvestigate PFS based on the $\alpha$-spending function $g_{\PFS}(\cdot,\alpha)$ after propagating all the level to this hypothesis. However, this is only of minor interest here, as $H_{0,\OS}$ has already been rejected. If \eqref{eq:design_I_joint_decision_OS} does not hold, the trial finishes without rejection of any hypothesis (Case 2.2).

\subsection{Including $\alpha$-spending for OS in the first analysis}\label{subsec:with_early_os}

\begin{figure}
	\centering
	\includegraphics[width=.8\textwidth]{"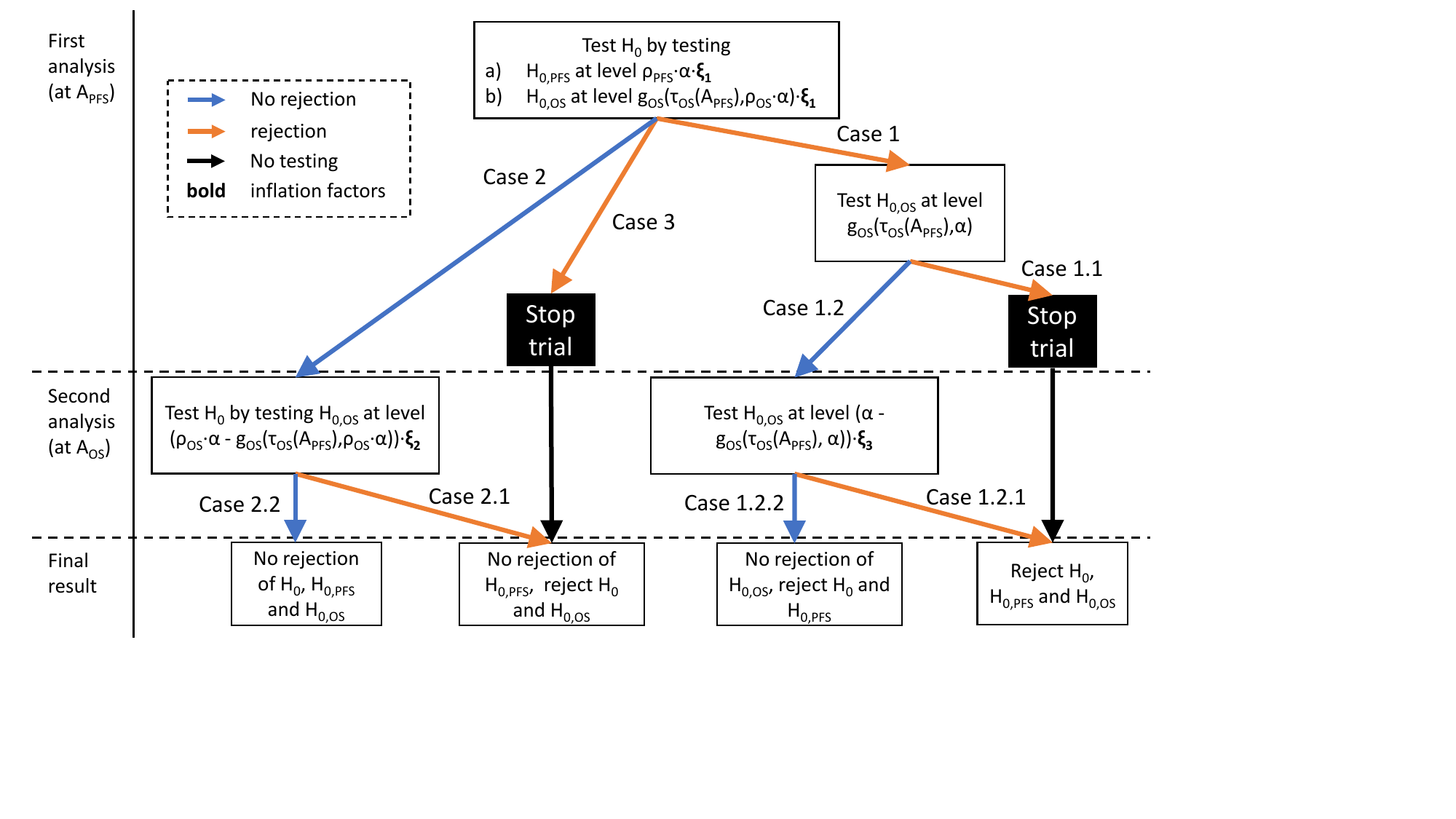"}
	\caption{Flow chart showing the course of a trial with early assessment of the null hypothesis for OS. The trial is stopped as soon as superiority regarding OS is shown.}
	\label{figure:flow_graph_II}
\end{figure}

In the first scenario above we only considered the option of testing OS at the first analysis if $\Hglob$ was already rejected based on PFS. However, as also suggested in \cite{Erdmann:2025}, we could also directly include a test for OS in the first analysis. In this case we have to choose $g_{\OS}(\cdot,\rho_{\OS}\alpha)$ in such a way that $g_{\OS}(\tau_{\OS}(A_{\PFS}),\rho_{\OS}\alpha) > 0$. As in \cite{Erdmann:2025}, one could e.g. choose a spending function that approximates the O'Brien-Fleming stopping boundaries (see \cite{Lan:83}) which is given by
\begin{equation}\label{eq:design_II_initial_spending}
	g_{\OS}(s,\rho_{\OS}\alpha) = 2 \cdot \left(1 - \Phi\left( \frac{\Phi^{-1}(1 - \rho_{\OS}\alpha/2)}{\sqrt{s}} \right) \right).
\end{equation}
or its one-sided version, respectively. After a potential propagation of the significance level, one would again have the choice of whether to use the additional level directly for the interim analysis or only in the final analysis. This would correspond to the choices
\begin{align}
	&g_{\OS}(s,\alpha) = \mathbbm{1}_{s \geq \tau_{\OS}(A_{\PFS})} \rho_{\PFS}\alpha + 2 - 2 \cdot \Phi\left( \frac{\Phi^{-1}(1 - \rho_{\OS}\alpha/2)}{\sqrt{s}} \right) \label{eq:design_II_full_spending_early} \\
	\text{and} \quad &g_{\OS}(s,\alpha) = \mathbbm{1}_{s \geq 1} \rho_{\PFS}\alpha + 2 - 2 \cdot \Phi\left( \frac{\Phi^{-1}(1 - \rho_{\OS}\alpha/2)}{\sqrt{s}} \right), \label{eq:design_II_full_spending_late}
\end{align} 
respectively.
We leave $g_{\PFS}$ as it was in the preceding subsection. Similarities and differences to the slightly more simple design of Section \ref{subsec:no_early_os} can already be seen by comparing Figures \ref{figure:flow_graph_I} and \ref{figure:flow_graph_II}. In what follows, we focus in particular on the differences to the prior design.
\paragraph{First analysis} 
The first major difference already occurs at analysis at the first analysis cutoff $A_{\PFS}$. We already want to assess the null hypothesis for OS when assessing $\Hglob$. In total, we are willing to spend a significance level of $\alpha_1 \coloneqq \rho_{\PFS}\alpha + g_{\OS}(\tau_{\OS}(A_{\PFS}),\rho_{\OS}\alpha)$. Analogously to \eqref{eq:design_I_inefficiency}, we will not exhaust this level if we test $H_{0,\PFS}$ at level $\rho_{\PFS}\alpha$ and $H_{0,\OS}$ at level $g_{\OS}(\tau_{\OS}(A_{\PFS}),\rho_{\OS}\alpha)$. As we conduct the two analyses simultaneously, we can now compute a joint inflation factor $\xi_1$ that solves
\begin{equation}\label{eq:gs_first_analysis}
	1 - \mathbb{P}_{\Hglob}\left[ \underbrace{\frac{U_{\PFS}(A_{\PFS})}{\sqrt{\hat{\sigma}_{\PFS}^2 (A_{\PFS})}} > \Phi^{-1}(\xi_1 \cdot \rho_{\PFS}\alpha) \,\bigcap\, \frac{U_{\OS}(A_{\PFS})}{\sqrt{\hat{\sigma}_{\OS}^2 (A_{\PFS})}} > \Phi^{-1}(\xi_1 \cdot g_{\OS}(\tau_{\OS}(A_{\PFS}),\rho_{\OS}\alpha))}_{\eqqcolon \Gamma(\xi_1)} \right] = \alpha_1.
\end{equation}
For the following tests we distinguish between three different cases. If we reject $H_{0,\PFS}$ at the inflated level $\xi_1 \cdot \rho_{\PFS}\alpha$ (Case 1), we proceed as in Case 1 of the previous subsection, either with an early rejection of $H_{0,\OS}$ (Case 1.1) and a resulting early termination of the trial or without (Case 1.2). If neither $H_{0,\PFS}$ nor $H_{0,\OS}$ are rejected at the inflated levels $\xi_1 \cdot \rho_{\PFS}\alpha$ and $\xi_1 \cdot g_{\OS}(\tau_{\OS}(A_{\PFS}),\rho_{\OS}\alpha)$, respectively (Case 2), we proceed to the second analysis as in Case 2 of the previous subsection. A small difference between the two scenarios is discussed below in the part on the second analysis. If $H_{0,\OS}$ is significant at the inflated level $\xi_1 \cdot g_{\OS}(\tau_{\OS}(A_{\PFS}),\rho_{\OS}\alpha)$ (Case 3), we can reject $\Hglob$ and also $H_{0,\OS}$ if $g_{\OS}(\tau_{\OS}(A_{\PFS}),\alpha) \geq \xi_1 \cdot g_{\OS}(\tau_{\OS}(A_{\PFS}),\rho_{\OS}\alpha)$ which ensures consonance of the testing procedure. In this case, we are inclined to stop the trial as we are able to reject the null hypothesis for our most important endpoint.

\paragraph{Second analysis}
As above, the analysis takes place at the random analysis cutoff date $A_{\OS}$. Case 1.2 is completely analogous to Case 1.2 of Subsection \ref{subsec:no_early_os}. In Figure \ref{figure:flow_graph_II}, the corresponding inflation factor is given by $\xi_3$. Case 2 here is a little bit different from the preceding Case 2 in Subsection \ref{subsec:no_early_os}.\\

\underline{Case 2:} We can still inflate when assessing the null hypothesis for OS within the intersection hypothesis $\Hglob$. In comparison to Section \ref{subsec:no_early_os}, however, we have already carried out a test for OS and an inflation of the local levels has already been carried out at the interim analysis. In dependence of the previously chosen inflation factor $\xi_1$, the continuation region for the intersection hypothesis has been defined as $\Gamma(\xi_1)$. At this analysis, we want to spend the remaining level $g_{\OS}(1, \rho_{\OS}\alpha) - g_{\OS}(\tau_{\OS}(A_{\PFS}), \rho_{\OS}\alpha)$. In order to exhaust this, we can compute the inflation factor $\xi_2$ that fulfills
\begin{equation*}
	1 - \mathbb{P}\left[ \Gamma(\xi_1) \; \bigcap \; \frac{U_{\OS}(A_{\OS})}{\sqrt{\hat{\sigma}_{\OS}^2 (A_{\OS})}} > \Phi^{-1}(\xi_2 \cdot (g_{\OS}(1, \rho_{\OS}\alpha) - g_{\OS}(\tau_{\OS}(A_{\PFS}), \rho_{\OS}\alpha))) \right] = \alpha
\end{equation*}
where we plug in the previously chosen inflation factor $\xi_1$. As earlier, we can compute this probability based on the consistent estimation of the covariance matrix of the three involved test statistics.

\section{Simulation studies}\label{sec:simulation_studies}
After our theoretical derivations we now empirically investigate the properties of our trial designs, most specifically whether they maintain the family-wise error rate (FWER). Also, we want explore how the power compares to designs that only use very simple corrections for the multiple testing problem or avoid it altogether by testing only one endpoint.\\
To properly account for the dependence between PFS and OS we use time-homogeneous Markovian multi-state models as in \cite{Erdmann:2025, Meller:2019} and adaptations thereof. The basic model consists of three different states. The current state of disease at time $s$ after trial entry is given by $X_i(s) \in \{0,1,2\}$. At trial entry, each patient starts in the initial state (0) and might then transition to the state of progressive disease (1) or to the state of death (2). After a transition to the progredient state, the patient can also die. Under the Markov assumption, the probabilities of transitions in the future only depend on the current state of the patient and are independent from further information about the previous course of disease. Then, these transition probabilities are governed by the transition intensities which are given by
\begin{equation*}
	\lambda_{kl}(s)\coloneqq \lim_{h \searrow 0} \frac{\mathbb{P}[X(s+h) = l|X(s) = k]}{h}
\end{equation*}
for $(k,l) \in \{(0,1), (0,2), (1,2)\}$. In a time-homogeneous model, these functions of time since trial entry $s$ are assumed to be constant. The values for the four baseline models we are considering here are shown in Table \ref{table:idm_parameters}. To assess type I error we generate scenarios in which patients of both treatment arms follow the transition intensities $\lambda^C_{kl}$ from Table \ref{table:idm_parameters}. To evaluate deviations from the Markov property we consider frailty models in which the transition intensities are multiplied by patient-specific random variables. This corresponds to a random rescaling of time, as shown in \cite{Aalen:1988}. We choose Gamma(10, 1/10) as the frailty distribution. For a meaningful statement about the influence of breaking the Markov assumption on the FWER, the same simulated data is used in the simulations with frailty as in the simulations without frailty, in that only an individual rescaling of the time is carried out (see \cite{Aalen:1988}). In order to assess the asymptotic behaviour of the testing procedures, we consider total sample sizes of $n \in \{128, 256, 640, 960, 1600\}$ patients, that are recruited over a period of 32 time units with an allocation ratio of $1/2$. The first analysis will take place as soon as a proportion of $r_{\PFS} = 25/64$ of these patients have experienced a PFS event. The second analysis will be conducted after $r_{\OS} = 38/64$ have died. We also simulate loss to follow-up by an independent, exponentially distributed variable with parameter $-\log(1-0.1)/12$.\\
To compare the power between different approaches, we choose transition intensities of the form
\begin{equation}\label{eq:alternatives_by_weighting}
	\lambda_{kl}^E = \lambda_{kl}^C - w \cdot (\lambda_{kl}^C - \lambda_{kl, \text{power}}^E)
\end{equation}
for $w \in \{0.6, 0.7, 0.8, 0.9, 1\}$ in the experimental arm without any consideration of frailty. Hence, for $w=1$ we obtain the scenarios of \cite{Erdmann:2025} for which the event number were tuned in such a way to achieve a power of 80\% to reject $H_{0,\PFS}$ and a power of 80\% to reject $H_{0,\OS}$ in the Bonferroni-adjusted design without an early analysis of OS-data.\\
\begin{table}[h]
	\centering
	\begin{tabular}{c||c|c|c|c|c}
		\multirow{2}{*}{Model} & $\lambda^C_{01}$ & $\lambda^C_{02}$ & $\lambda^C_{12}$ & \multirow{2}{*}{$\lceil r_{\PFS} \cdot n \rceil$} & \multirow{2}{*}{$\lceil r_{\OS} \cdot n \rceil$}\\
		& $\lambda^E_{01,\text{power}}$ & $\lambda^E_{02,\text{power}}$ & $\lambda^E_{12,\text{power}}$ & & \\
		\hline
		\hline
		\multirow{2}{*}{1} & 0.06 & 0.30 & 0.30 & \multirow{2}{*}{433} & \multirow{2}{*}{630}\\
		& 0.10 & 0.40 & 0.30 & & \\
		\hline
		\multirow{2}{*}{2} & 0.30 & 0.28 & 0.50 & \multirow{2}{*}{452} & \multirow{2}{*}{747}\\
		& 0.50 & 0.30 & 0.60 & & \\
		\hline
		\multirow{2}{*}{3} & 0.140 & 0.112 & 0.250 & \multirow{2}{*}{644} & \multirow{2}{*}{742}\\
		& 0.180 & 0.150 & 0.255 & & \\
		\hline
		\multirow{2}{*}{4} & 0.18 & 0.06 & 0.17 & \multirow{2}{*}{940} & \multirow{2}{*}{963}\\
		& 0.23 & 0.07 & 0.19 & & \\
	\end{tabular}
	\caption{Parameter configurations for the time-homogeneous Markovian illness-death models considered in our simulations and number of events at which the two analyses are triggered.}
	\label{table:idm_parameters}
\end{table}%
In each simulation run, up to 800 patients per treatment arm will be recruited, with a rate of 25 patients per arm per time unit. As above, recruitment stops as soon as $A_{\OS}$ is reached which might occur earlier. In the scenarios investigating the power of the approaches, 25 patients are recruited per time unit and per arm. As above, loss to follow-up is simulated by an independent, exponentially distributed variable with parameter $-\log(1-0.1)/12$. The interim analysis and final analysis are triggered by the number of observed PFS and OS events, respectively, which are given in the last two columns of Table \ref{table:idm_parameters}. These event numbers are chosen so that in the case $w=1$ in \eqref{eq:alternatives_by_weighting}, the power to reject $H_{0,\PFS}$ and the power to reject $H_{0,\OS}$ are both 80\%. In \cite{Erdmann:2025} it is described in more detail how these were derived by simulation.\\
For all those scenarios, we will compare nine different testing approaches that are described in further detail in Table \ref{table:testing_procedures}.  The group of the first four is designed so that no test for overall survival is planned in the interim analysis (as in subsection \ref{subsec:no_early_os}). In the group of the next four, OS is always assessed in the first analysis (as in subsection \ref{subsec:with_early_os}). The corresponding critical values are determined by the alpha-spending function according to O'Brien-Fleming. Within both of the two groups mentioned above, we first consider a Bonferroni correction, which ensures that PFS is tested at the one-sided level of 0.005 and OS at the one-sided level of 0.02. As a first improvement, we also consider a testing procedure, which recycles the corresponding significance level if one of the two hypotheses can be rejected \citep{Maurer:2013}. Finally, we use the procedures presented in Section \ref{sec:sequential_designs} to try to exhaust the family-wise error rate. Within the closed testing procedure, we consider both the option to use the propagated significance level only in the final analysis or to use it already in the interim analysis. Finally, we also consider the option to conduct a single test for OS only in the final analysis at full significance level. This can serve as a benchmark as it should give the highest overall power to reject $H_{0,\OS}$. In comparison with other methods, we are primarily interested in how much power is lost because of the Bonferroni-correction and what proportion of this can be made up by improved test procedures.\\
We compare various measures between these 9 procedures. These are the empirical rejection proportions of $H_{0,\PFS}$, $H_{0,\OS}$ of at least one of these hypotheses and of both hypotheses simultaneously. Under the null hypothesis, the proportion of rejections of at least one hypothesis is our estimate of the family-wise error rate. Under alternatives, this value is the disjunctive power and the frequency of simultaneous rejections of both hypotheses is the conjunctive power.\\
\begin{table}[h]
	\centering
	\begin{tabularx}{\textwidth}{l|X}
		Abbreviation  &  Description\\ 
		\hline 
		\textbf{BON} & Bonferroni-adjusted testing procedure with a test of PFS at level $\rho_{\PFS}\alpha$ at the interim analysis and a test of OS at level $\rho_{\OS}\alpha$ at the interim analysis.\\
        \textbf{REC} & Bonferroni-adjusted testing as above with recycling of $\rho_{\PFS}\alpha$ after rejection of $H_{0,\PFS}$.\\
		\textbf{EX/LAST} & Improved closed testing procedure with exhaustion of the FWER and $\alpha$-spending functions as in \eqref{eq:design_I_initial_spending} and \eqref{eq:design_I_full_spending_late}.\\ 
		\textbf{EX/FIRST} & Improved closed testing procedure with exhaustion of the FWER and $\alpha$-spending functions as in \eqref{eq:design_I_initial_spending} and \eqref{eq:design_I_full_spending_early}.\\
		\textbf{BON/GS} & Bonferroni-adjusted testing procedure with a test of $H_{0,\PFS}$ at level $\rho_{\PFS}\alpha$ at the interim analysis and a group-sequential procedure for OS at level $\rho_{\OS}\alpha$ with $\alpha$-spending function as in \eqref{eq:design_II_initial_spending}.\\
        \textbf{REC/GS} & Bonferroni-adjusted testing as above with recycling significance level if one of the hypotheses is rejected at the interim analysis.\\
		\textbf{EX/GS/LAST} & Improved closed testing procedure with exhaustion of the FWER and $\alpha$-spending functions as in \eqref{eq:design_II_initial_spending} and \eqref{eq:design_II_full_spending_late}.\\ 
		\textbf{EX/GS/FIRST} & Improved closed testing procedure with exhaustion of the FWER and $\alpha$-spending functions as in \eqref{eq:design_II_initial_spending} and \eqref{eq:design_II_full_spending_early}.\\
		\textbf{OS} & Only one test of $H_{0,\OS}$ at the final analysis at level $\alpha$. 
	\end{tabularx}
	\caption{Overview of the different testing procedures.}
	\label{table:testing_procedures}
\end{table}%
Each scenario is simulated 100,000 times. For a true underlying value of 0.025 and 0.8, the Monte Carlo estimates of our simulations will hence lie within the intervals $[0.240, 0.260]$ and $[0.7975, 0.8025]$, respectively, with a probability of 95\%.

\subsection{Results}
At first, we want to check whether the improved testing procedures control the FWER in Markovian and also non-Markovian settings. In Figure \ref{figure:fwer_comparisons}, we compare FWERs for the three testing approaches BON, EX/LAST and OS with and without frailty modeling in all four scenarios in dependence of the total sample size. For the sake of clarity, we do not show the values for the other strategies from Table \ref{table:testing_procedures} here. It is shown in the Supplementary Material that these are equivalent or very similar to those of the strategies BON or EX/LAST, respectively.
\begin{figure}[h]
	\centering
	\includegraphics[width=\textwidth]{"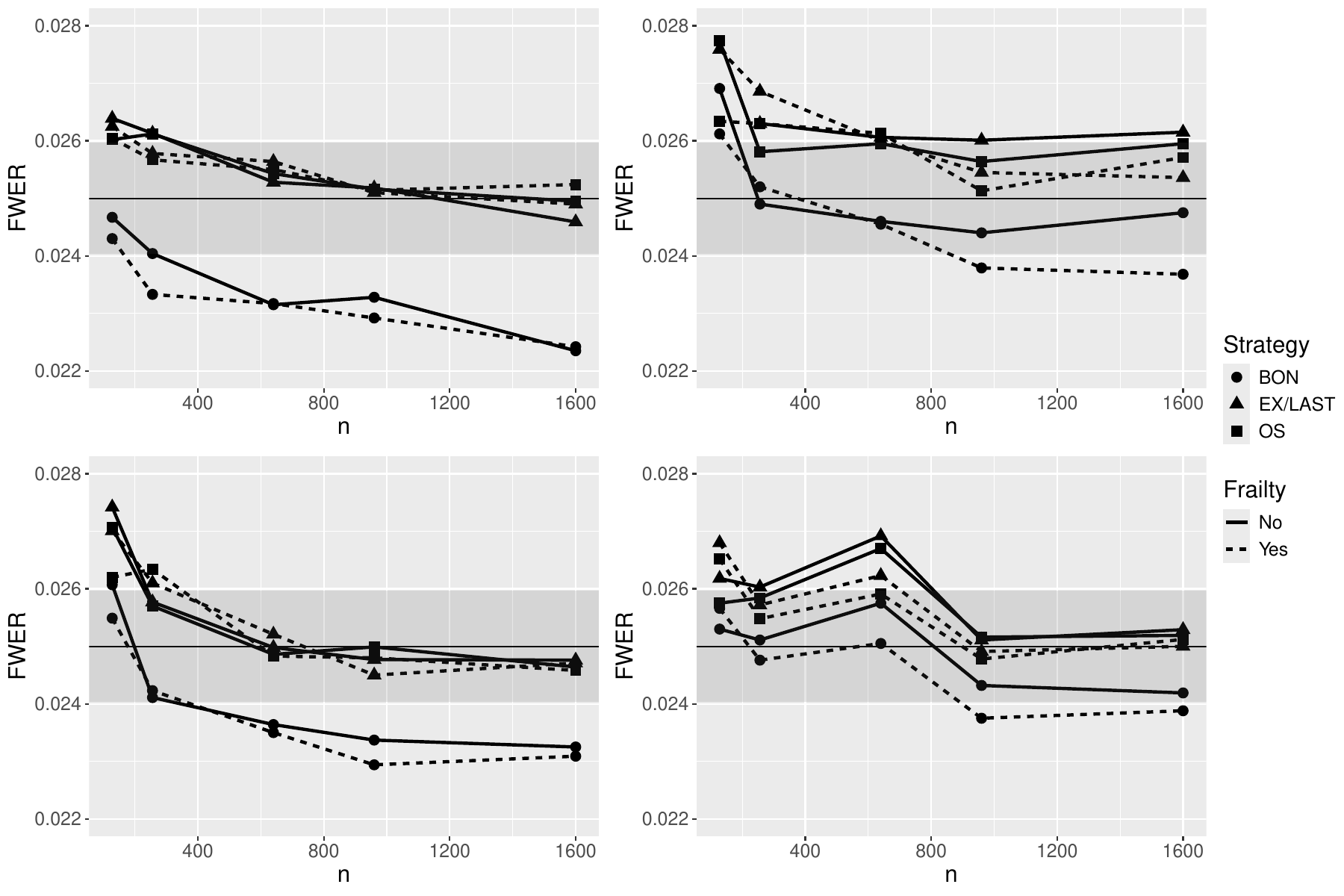"}
	\label{figure:fwer_comparisons}
	\caption{FWERs of the three testing procedures BON, EX/LAST and OS with and without consideration of frailties in all four scenarios. The shaded area characterises the Monte Carlo sampling error interval. The subfigures are arranged as follows: Scenario 1, top left; Scenario 2, top right; Scenario 3, bottom left; Scenario 4, bottom right.}
\end{figure}%
The following observations can be made throughout all scenarios: For small sample sizes, all of the approaches that should exhaust the nominal type I error rate (EX/LAST and OS) are slightly anti-conservative. The empirical rejection rates of the new procedure do not substantially exceed those of the simple test for OS, whose anti-conservativeness is well known for small numbers of cases \citep{Kellerer:1983}. The slight inflation can therefore possibly be attributed to similar problems. For moderate and large sample sizes, the Bonferroni-corrected procedure is clearly conservative as the dependency between the test statistics is not exploited. The other two procedures exhaust the nominal level of 2.5\% without noticeably exceeding it. In all cases, there are no relevant differences between simulations with and without frailty. We interpret this as the new procedures are not sensitive to the violation of the Markov assumption.\\
In Table \ref{table:results_scenario1_w=1}, we summarise results for the nine approaches under the parameter configurations of Scenario 1 of Table \ref{table:idm_parameters} with $w=1$ in \eqref{eq:alternatives_by_weighting}.
\begin{table}[h]
	\centering
	\begin{tabular}{l || c c c c c}
		Testing procedure & Rej. $H_{0,\PFS}$ & Rej. $H_{0,\OS}$ & Disj. power & Conj. power & Early stop\\
		\hline
		BON & 0.7937 & 0.8072 & 0.8960 & 0.7049 & 0.0000 \\
        REC & 0.7937 & 0.8225 & 0.8960 & 0.7202 & 0.0000 \\
		EX/LAST & 0.7937 & 0.8264 & 0.8999 & 0.7202 & 0.0000 \\
		EX/FIRST & 0.7937 & 0.8262 & 0.8999 & 0.7200 & 0.4265 \\
		BON/GS & 0.7937 & 0.8067 & 0.8958 & 0.7046 & 0.1396 \\
        REC/GS & 0.7937 & 0.8220 & 0.8958 & 0.7200 & 0.1730 \\
		EX/GS/LAST & 0.7970 & 0.8263 & 0.9007 & 0.7226 & 0.1396 \\
		EX/GS/FIRST & 0.7970 & 0.8258 & 0.9007 & 0.7222 & 0.4326 \\
		OS & 0.0000 & 0.8313 & 0.8313 & 0.0000 & 0.0000 
	\end{tabular}
	\caption{Empirical rejection rates for the different multiple testing procedures}
	\label{table:results_scenario1_w=1}
\end{table}%
The first eight procedures behave similarly when assessing $H_{0,\PFS}$. We find a slight advantage for the two new procedures that already assess $H_{0,\PFS}$ at the interim analysis. This, because we inflate the level for this test according to \eqref{eq:gs_first_analysis}. The rejection rate for $H_{0,\OS}$ suffers from a drop of about 2.4 percentage points when using the Bonferroni-corrected designs compared with the design in which only OS-data is tested. More than 60\% of this loss can be recovered by recycling the significance level allocated to hypotheses that can be rejected at the first analysis. By exploiting the joint distribution of the test statistics, even more than 75\% of this loss can be recovered by using one of the newly proposed procedures. By construction, the disjunctive power of the Bonferroni-corrected procedures and those that only recycle the significance level is the same. The new procedures also show an increase of about 0.5 percentage points with respect to this measure. Hence, this increase is only due to the exploitation of the dependence structure of the test statistics as the benefits of the graph-based closed testing procedure only take effect if one hypothesis could already be rejected. The conjunctive power is increased by about 1.5 percentage points compared to the Bonferroni-corrected designs. For the group-sequential design and especially for the designs that recycle significance level for OS directly after $H_{0,\PFS}$ has been rejected (i.e. those that use spending functions as in \eqref{eq:design_I_full_spending_early} and \eqref{eq:design_II_full_spending_early}), there is also a noticeably large probability of stopping the trial early for success when all involved hypotheses have been rejected.\\
Similar effects can be observed if the parameter $w$ in \eqref{eq:alternatives_by_weighting} is varied within the configurations of Scenario 1 in order to consider different alternatives. As above, we are particularly interested in differences in rejection proportions of $H_{0,\OS}$ and in differences in disjunctive power of the new procedures compared to the Bonferroni-corrected procedure. In Figure \ref{figure:power_comparisons} the relative difference for both quantities compared to the corresponding Bonferroni-corrected procedure are shown. For the sake of clarity we do not show values for all procedures because EX/FIRST and EX/GS/FIRST as well as EX/LAST and EX/GS/LAST perform similarly in terms of rejection proportions of $H_{0,\OS}$. Also, EX/FIRST and EX/LAST as well as EX/GS/FIRST and EX/GS/LAST are very similar in terms of disjunctive power. Furthermore, this quantity is always the same for the Bonferroni-corrected procedures and the improved procedures that can potentially recycle significance level. See Table \ref{table:results_scenario1_w=1} to get an impression of these circumstances. Gains in disjunctive power are even larger for smaller effect sizes. This holds for the comparison with the Bonferroni-corrected procedure as well as for the comparison with methods that can recycle some significance level. Of course, these gains diminish with increasing effect size as all procedures then approach a power of close to 100\%. Across all the different values for $w$ considered here, the improved procedures compensate about 2/3 of the power lost due to the Bonferroni correction for the test of $H_{0,\OS}$.
\begin{figure}
	\centering
	\includegraphics[width=\textwidth]{"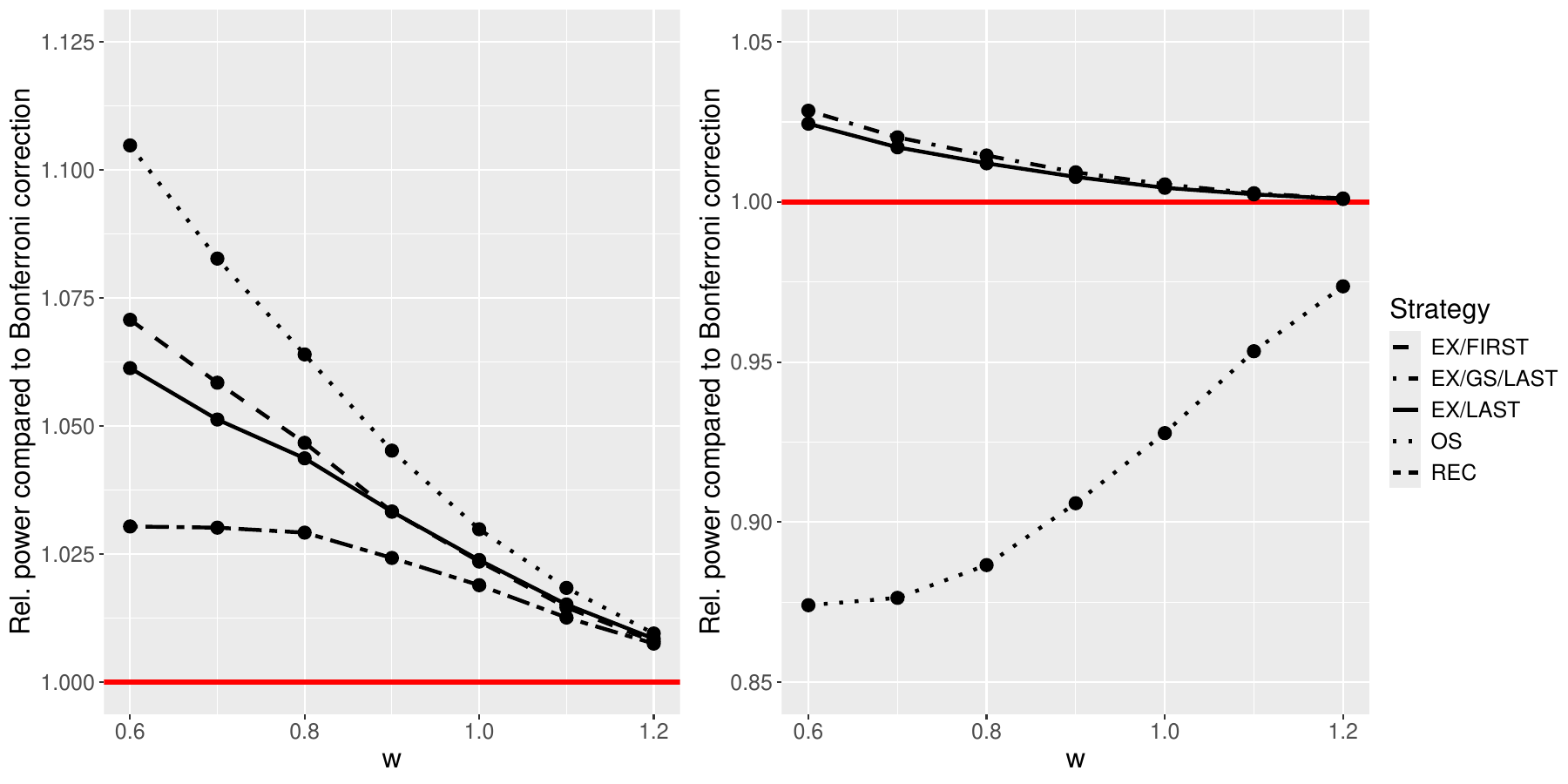"}
	\label{figure:power_comparisons}
	\caption{Left panel: Relative differences in power to reject $H_{0,\OS}$ of improved testing procedures and the OS procedure compared to the BON procedure. Right panel: Relative differences in disjunctive power of improved testing procedures and the OS procedure compared to the BON procedure.}
\end{figure}%
The other scenarios determined by Table \ref{table:idm_parameters} yield similar results and are shown in the Supplementary Material.

\section{Practical issues}\label{sec:practical_issues}

\subsection{Choice and practical use of $\alpha$-spending functions}
In a group-sequential clinical trial the alpha-spending function has to be chosen at the design stage. In our experience, if we only have one endpoint, the most prevalent choice is an O'Brien-Fleming spending function, or the Lan-DeMets approximation \citep{Lan:83} to it. This, because this alpha-spending function distributes the significance level such that early rejection is less likely. In clinical development this is often desired because the smaller amount of information for a benefit-risk assessment at an early time point is balanced, in case of early stopping, by a large observed effect.
How to assess PFS and OS in a clinical trial with interim analysis such that FWER is protected by a hierarchical multiple testing procedure has been discussed in detail in \cite{Glimm:2010}. Interest is typically in, what the authors call, an “overall hierarchical” strategy: Here, the secondary endpoint, say OS, is 
only assessed if the null hypothesis for the primary endpoint PFS is rejected. In addition, in case of non-rejection of OS it will also be assessed at further pre-specified analysis. \cite{Glimm:2010} show that FWER is only maintained if for both endpoints a group-sequential procedure is used. In this case – how should alpha-spending functions be chosen? Considerations for both the primary and secondary endpoint in this scenario are no different than described above in case of only one endpoint. The same $\alpha$-spending functions may be used for both endpoints, or alternatively, a spending function with a larger significance level at the interim analysis for the secondary endpoint. This, in order to increase the probability of stopping the entire trial early. We refer to \cite{Hung:07} as well as \cite{Tamhane:10, Tamhane:18} for further discussions.

\subsection{How to handle cutoff prediction uncertainty}
This aspect has been discussed in detail in \cite{Asikanius:2024}. For completeness we recap that discussion here in slightly abbreviated form. In group-sequential trials with a time-to-event endpoint capturing of events is not instantaneous. For example, a progression event in an oncology trial is typically not entered on the day it was detected by the treating physician, but later when the center enters the data in the trial database in batches. Further delay happens because of data cleaning and potential event adjudication. The sponsor or its data monitors checks the key data for plausibility, consistency and correctness, so that data might be subject to change, including changes to dates of events, or even the addition or removal of events. Furthermore, in large multinational trials, prospective planning and communication of timelines is required because a large number of individuals are responsible for day-to-day trial conduct. A date for the clinical cutoff date is therefore predicted based on the past occurrence of events. Because of variability (e.g. in how events happen) the number of observed events by that date will differ from the predicted, targeted number of events. When a snapshot of the cleaned database will be taken it typically happens that we do not precisely meet the targeted number of events. This results in an information fraction which is lower or higher than planned. Significance levels computed at the design stage based on the assumed information fractions are therefore recalculated according to the observed information fraction, where information fractions remain relative to the target number of events planned for the primary analysis. Recalculation is done using the $\alpha$-spending approach introduced by \cite{Lan:83}.

\subsection{Consonance of the testing procedure}
A closed testing procedure is called consonant if the rejection of the global null hypothesis leads to a rejection of at least one elementary hypothesis. Not only for the sake of interpretability of the results, consonance is a desirable property.\\
In our case, we only need to make sure that the rejection of $H_{0,\text{global}}$ also implies rejection of at least one of $H_{0,\PFS}$ and $H_{0,\OS}$. Although this seems clear at first glance, it is possible that the choice of the $\alpha$-functions $g_{\PFS}$ and $g_{\OS}$ may lead to non-consonant decisions. For the design in Section \ref{subsec:no_early_os}, consonance is achieved if and only if
\begin{enumerate}[label=(\roman*)]
    \item $\rho_{\PFS} \alpha \leq g_{\PFS}(\tau_{\PFS}(A_{\PFS}), \alpha) \quad$ and 
    \item $\xi_2 \cdot \rho_{\OS}\alpha \leq \xi_1 \cdot (\alpha - g_{\OS}(\tau_{\OS}(A_{\PFS}), \alpha))$
\end{enumerate}
are given (see Figure \ref{figure:flow_graph_I} as a reference for the inflation factors applied here). Similar conditions arise for the slightly more complex design of Section \ref{subsec:with_early_os}.\\
As noted by \cite{Anderson:2022}, a sufficient condition for consonance would be
\begin{align*}
    \frac{g_{E}(\cdot, \rho_{E} \alpha)}{g_{E}(\cdot, \alpha)} \equiv \rho_{E}
\end{align*}
for both $E \in \{\PFS,\OS\}$. These in particular apply the two conditions mentioned above. For sufficient conditions for consonance in this and and other testing procedures that might involve more than two endpoints, we refer to \cite{Anderson:2022}.

\subsection{Planning of the trial}
The power gains demonstrated in our simulations of Section \ref{sec:simulation_studies} can of course also be translated into a smaller number of events required to achieve the desired power. In \cite{Erdmann:2025}, a simulation routine was set up to determine these number of events. For the simplest design of Section \ref{subsec:no_early_os}, where OS is only tested at the interim analysis if PFS could have been rejected, the targeted number of PFS events for the interim analysis does not deviate from the numbers found in \cite{Erdmann:2025} as no additional inflation is possible in this scenario. On the contrary, the required number of OS events to achieve the desired power of 80\% to reject $H_{0,\OS}$ can be lowered when taking the potential recycling of $\rho_{\PFS}\alpha$ and the inflation of the new procedures into account. For the scenarios of Table \ref{table:idm_parameters}, the targeted number of OS events reduces to 594, 718, 703 and 919, respectively. Hence, the number can be reduced by approximately 5\% across all our scenarios.\\
We would like to emphasise once again that the specification of the entire illness-death model, as well as the assumptions about the recruitment process and possible loss to follow-up, are decisive for the overall power planning. In particular, it is possible that two different illness-death models could lead to similar marginal distributions for PFS and OS, but different dependence structures and thus correlations between the test statistics.

\section{Discussion}\label{sec:discussion}
In this paper, we propose a testing procedure for trials with endpoints that can be embedded in an illness-death multi-state model. We exemplify this for the two typical oncology endpoints, PFS and OS. Our contribution is that we fully exhaust the FWER within the multiple testing problem, resulting in power gains, in particular when both interim and final analysis are event-driven. The joint distribution of log-rank test statistics in group sequential designs from \cite{Lin:1991} is combined with a closed testing procedure that does not only address a global null hypothesis. Furthermore, we consider this in the context of event-driven censoring. For individual log-rank tests, \cite{Rühl:2023} investigated the necessary property of independent censoring for finite samples. Since we are also interested in event-driven censoring across endpoints, we use an asymptotic approach in our theoretical foundations. On this basis, we can apply the powerful framework from \cite{Anderson:2022}. Going beyond the examples mentioned there, we apply it to asymptotically normally distributed tests whose covariance structure is unknown but consistently estimable according to our theoretical results.\\
Generalisations to more than two endpoints and more than two analysis cutoffs are generally feasible and will be elaborated upon in future research. More than two endpoints might be of interest when even more events that characterise the course of disease shall be incorporated into the analysis as e.g. in the multi-state models presented in \cite{Le-Rademacher:2018, Danzer:2024}. It is also possible to perform more than one interim analysis. However, this involves an increased risk of subsequent rejection of hypotheses based on past analyses. These may then not be supported by the current data.\cite{Glimm:2010} discussed that, in practice, one should refrain from doing so, even at the expense of a slight loss of power, in order not to jeopardise the explainability of the results. Great caution must also be exercised when considering adaptations of the trial design at interim analyses as e.g. sample size recalculations. This is generally not possible here, as PFS acts as a surrogate for OS and improper use of this information can lead to an inflation of the FWER. Solutions to this problem have so far only been found by discarding some information, using worst-case estimates \citep{Magirr:2016} or making additional assumptions about the relationship between the endpoints \citep{Danzer:2024}.\\
So far, we focused on developing methods for hypothesis testing. \cite{Zhao:2025} presented the computation of $p$-values for the framework of \cite{Anderson:2022} that we used here. Reporting of effect measures remains challenging. \cite{Izumi:2025} investigated the bias of the estimates of hazard ratios for OS in the hierarchical design of \cite{Glimm:2010} and proposed unbiased estimates. Extensions to our more flexible designs are certainly possible. However, as pointed out by \cite{Erdmann:2025}, proportional hazards for both endpoints simultaneously appear quite implausible. As an alternative, the three transitions of the illness-death model could be targeted separately (as e.g. in \cite{Le-Rademacher:2018}). Nevertheless, the possible bias of those estimates that are caused by the sequential nature of the procedure also needs to be addressed.\\ 
The challenges of non-proportional hazards can also be considered when choosing the test statistics. It is well-known that the standard log-rank test that is applied here for both endpoints is semi-parametrically optimal for proportional hazards. However, this optimality is lost when other types of deviations are present. If the type of deviation (e.g. early or late separation of survival functions) is known, a correspondingly weighted log-rank test (e.g. from the family of weights introduced by Fleming and Harrington in \cite{Harrington:1982}) can be chosen. If this cannot be anticipated in advance, combination tests as the 'max-combo' \citep{Lee:2007} or the 'mdir' test \cite{Brendel:2014} can be applied. Although we did not account for weighted tests specifically in this manuscript, an extension is generally possible. However, the correlation structure now becomes even more complex, as not only the correlation across endpoints and analysis time points must be taken into account, but also across several weighted test statistics for the individual endpoints.
Our simulation studies in Section \ref{sec:simulation_studies} reveal that the adherence to the nominal type I error of our proposed procedures suffers for sample sizes below 500. This is analogous to the well-known anti-conservativeness of the standard log-rank test for small sample sizes \citep{Kellerer:1983}. 
\cite{Persson:2019} proposed a permutation-based approach for the one-stage testing method of \cite{Wei:1984}. An extension to group-sequential testing as in \cite{Lin:1991} and to our procedure should in general be possible. However, the crucial condition of exchangeability for such permutation procedures should be critically investigated within those efforts before applying it in practice.

\section*{Acknowledgements}
The work of MFD was funded by the Deutsche Forschungsgemeinschaft (DFG, German Research Foundation)–413730122. Parts of the calculations for this publication were performed on the HPC cluster PALMA II of the University of Münster, subsidised by the DFG (INST 211/667-1). JB acknowledges support by DFG research grant BE 4500/7-2. The authors would like to thank Marcel Wolbers for reviewing the manuscript.

\section*{Data and code availability}
The code and results of the simulation study can be accessed at \url{https://github.com/moedancer/MultSurvTrialDesign}.

\putbib
\end{bibunit}

\newpage

\appendix

\renewcommand{\thesection}{\Alph{section}}
\setcounter{figure}{0}
\setcounter{table}{0}
\renewcommand{\figurename}{Supplementary Figure}
\renewcommand{\tablename}{Supplementary Table}
\renewcommand{\thetable}{S\arabic{table}}
\renewcommand{\thefigure}{S\arabic{figure}}

\begin{bibunit}

\begin{center}
    \huge
    Supplementary Material for\\
    'Exhausting the type I error level in event-driven group-sequential designs with a closed testing procedure for progression-free and overall survival'
\end{center}

\vspace{12pt}

\section{Technical appendix}

At first, we consider the asymptotic behaviour of the event-driven analysis dates $A_{\PFS}$ and $A_{\OS}$. We define $\Delta_E\coloneqq \Delta_E(\infty)$ and $X_E\coloneqq X_E(\infty)$ for both events $E \in \{\PFS, \OS\}$. The analysis cutoffs of the respective events are given by
\begin{equation*}
	D_{\PFS}\coloneqq X_{\PFS} +  R \text{ and } D_{\OS}\coloneqq X_{\OS} +  R. 
\end{equation*}  
Let $F^{uc}_{D_{\PFS}}$ and $F^{uc}_{D_{\OS}}$ denote the subdistribution functions of the calendar dates $D_{\PFS}$ and $D_{\OS}$ that were not censored, i.e.
\begin{equation*}
	F^{uc}_{D_{\PFS}}(t) \coloneqq \mathbb{P}[D_{\PFS} \leq t; \Delta_{\PFS} = 1] \quad \text{and} \quad F^{uc}_{D_{\OS}}(t) \coloneqq \mathbb{P}[D_{\OS} \leq t; \Delta_{\OS} = 1].
\end{equation*}
We assume that those functions are continuous and we require $r_{\PFS}$ and $r_{\OS}$ to be in the interior of the images of $F^{uc}_{D_{\PFS}}$ and $F^{uc}_{D_{\OS}}$, respectively. The corresponding quantile functions $Q^{uc}_{D_{\PFS}}$ and $Q^{uc}_{D_{\OS}}$ are given by 
\begin{equation*}
	Q^{uc}_{D_{\PFS}}(p) \coloneqq \inf\left\{ t \geq 0 \colon F^{uc}_{D_{\PFS}}(t) \geq p\right\} \quad \text{and} \quad Q^{uc}_{D_{\OS}}(p) \coloneqq \inf\left\{ t \geq 0 \colon F^{uc}_{D_{\OS}}(t) \geq p\right\}.
\end{equation*}
Their empirical counterparts that are determined by the study sample are denoted by $\hat{F}^{uc}_{D_{\PFS}}$ and $\hat{F}^{uc}_{D_{\OS}}$ and respectively $\hat{Q}^{uc}_{D_{\PFS}}$ and $\hat{Q}^{uc}_{D_{\OS}}$. In this context, the analysis dates are given by
\begin{equation*}
	A_{\PFS}=\hat{Q}^{uc}_{D_{\PFS}}(r_{\PFS}) \quad \text{and} \quad A_{\OS}=\hat{Q}^{uc}_{D_{\OS}}(r_{\OS}).
\end{equation*}
The first theoretical result yields the convergence of these random calendar dates, at which the analyses are conducted, converge to deterministic values that are determined by the quantile functions given above.

\begin{lemma}\label{lemma:calendar_dates}
	The event-driven random calendar dates defined in \eqref{eq:caldate_analysis_pfs} and \eqref{eq:caldate_analysis_os}, respectively, both converge in probability to deterministic calendar dates $t_{\PFS}$ and $t_{\OS}$, defined in \eqref{eq:limits_analysis_dates}, i.e.
	\begin{equation*}
		A_{\PFS} \overset {\mathbb{P}}{\to} t_{\PFS} \quad \text{and} \quad A_{\OS} \overset {\mathbb{P}}{\to} t_{\OS}.
	\end{equation*}
\end{lemma}
\begin{proof}
	Let $E \in \{\PFS, \OS\}$. For a standard normally distributed random variable $Z$, we have
	\begin{equation*}
		\hat{F}^{uc}_{D_E}(Z) \overset{\text{a.s.}}{\to} F^{uc}_{D_E}(Z)
	\end{equation*}
	as a consequence  of the Glivenko-Cantelli theorem. For the standard normal distribution function, it follows that
	\begin{equation*}
		\Phi\left( \hat{Q}^{uc}_{D_E}(t_E) \right) = \mathbb{P}\left[\hat{F}^{uc}_{D_E}(Z) < t_E \right] \to \mathbb{P}\left[F^{uc}_{D_E}(Z) < t_E \right] = \Phi\left( Q^{uc}_{D_E}(t_E) \right)
	\end{equation*} 
	if $F^{-1}$ is continuous at $t_E$. By the continuity of $\Phi^{-1}$ it follows that $A_{\PFS} \overset {\text{a.s.}}{\to} t_E$ which in particular implies convergence in probability.
\end{proof}
Before stating the asymptotic distribution of the log-rank statistics at the event-driven analysis dates, we also require the following Lemma. For an arbitrary stopping time $A$, it is not clear that an analysis at this random calendar date is asymptotically equivalent to an analysis at the fixed calendar date to which this random date converges. Of course, well-known results (see e.g. \cite{Sellke:1983}) yield this result if the stopping time is given by a number of events of the same event that is tested. However, we also want to test $H_{0,\OS}$ at the analysis triggered by PFS events and vice versa. However, the key point of this proof is the characterization of the process of log-rank statistics in calendar time of \cite{Olschewski:1986} and the assumption of the continuity of the time-transformation that is applied to the Brownian motion. This continuity ensures that random fluctuations around the deterministic calendar date are asymptotically negligible.

\begin{lemma}\label{lemma:convergence_random_stopping_time}
	Let $(U(t))_{t\geq 0}$ be a stochastic process that converges in distribution to a time-changed Brownian motion with time-transformation $\phi$, i.e.
	\begin{equation*}
		(U(t))_{t \geq 0} \overset{\mathcal{D}}{\to} (W(\phi(t)))_{t \geq 0}
	\end{equation*}
	on the space of càdlàg functions $D([0,\tau])$ for a Brownian motion $W$ and an arbitrary $\tau > 0$. Moreover, let $A$ be a (positive) random variable with $A \overset{\mathbb{P}}{\to} t_0$ s.t. $\phi$ is continuous at $t_0$. Then we also have
	\begin{equation*}
		(U(A) - U(t_0)) \overset{\mathbb{P}}{\to} 0.
	\end{equation*} 
\end{lemma}
\begin{proof}
	It holds
	\begin{align*}
		&\mathbb{P}[|U(A) - U(t_0)| > \varepsilon]\\
		= & \mathbb{P}[|U(A) - U(t_0)| > \varepsilon; |A - t_0| < \gamma] + \mathbb{P}[|U(A) - U(t_0)| > \varepsilon; |A - t_0| \geq \gamma]\\
		\leq & \mathbb{P}\left[ \sup_{s\colon |s-t_0| \leq \gamma} |U(s) - U(t_0)| > \varepsilon \right] + \mathbb{P}[|A - t_0| \geq \gamma]
	\end{align*}
	for any $\gamma > 0$. For any fixed $\gamma$ the second summand vanishes as $A$ converges in probability to $t$ by our assumptions.\\
	For any fixed $\gamma$, we consider for the standard Brownian motion $W$ the probability
	\begin{align*}
		&\mathbb{P}\left[ \sup_{s\colon |s-t_0| \leq \gamma} |W(\phi(s)) - W(\phi(t_0))| > \varepsilon \right]\\
		=&2 \cdot \mathbb{P}[W(\phi(t_0 + \gamma) - \phi(t_0 - \gamma)) > \varepsilon]
	\end{align*}
	that is obtained by the reflection principle. By the continuity assumption on $\phi$, if $\gamma$ is small enough, this probability can get arbitrarily small. From the Portmanteau Theorem (see e.g. Lemma 2.2 (vii) in \cite{vanderVaart:2000}), we obtain 
	\begin{equation*}
		\mathbb{P}\left[ \sup_{s\colon |s-t_0| \leq \gamma} |U(s) - U(t_0)| > \varepsilon \right] \to 2 \cdot \mathbb{P}[W(\phi(t + \gamma) - \phi(t - \gamma)) > \varepsilon].
	\end{equation*}
	This concludes the proof because now, we can choose $\gamma$ small enough s.t. for some $\delta > 0$, the probabilities $\mathbb{P}[|A - t_0| \geq \gamma]$ and $2 \cdot \mathbb{P}[W(\phi(t + \gamma) - \phi(t - \gamma)) > \varepsilon]$ are both smaller than $\delta/3$ and we can choose $n$ large enough s.t. 
	\begin{equation*}
		\left|\mathbb{P}\left[ \sup_{s\colon |s-t_0| \leq \gamma} |W(\phi(s)) - W(\phi(t_0))| > \varepsilon \right] - 2 \cdot \mathbb{P}[W(\phi(t + \gamma) - \phi(t - \gamma)) > \varepsilon]\right| < \delta/3.
	\end{equation*}
\end{proof}
An alternative proof considers joint convergence in distribution of $(U, A)$, mapped onto $(U(A)-U(t_0))$. Replacing convergence in distribution with a.s. convergence of representations equal in distribution using the Skorokhod-Dudley almost sure representation theorem, one finds a.s. convergence of the difference of interest to zero. This implies convergence to zero in distribution, and by construction convergence in distribution of the original $(U(A)-U(t_0))$ to zero. The latter is the desired result, since convergence in distribution to zero is convergence in probability.\\
Now, we can state the asymptotic distribution of the log-rank statistics. We connect the preceding Lemma with standard arguments that were also applied in \cite{Tsiatis:1981, Lin:1991} to obtain the asymptotic distribution at the random analysis dates.

\begin{theorem}\label{thm:main_convergence}
	Under the strict null hypothesis of equal distribution of PFS and OS in both groups, the joint distribution of PFS and OS log-rank statistics evaluated at event-driven analysis dates $A_{\PFS}$ and $A_{\OS}$, respectively, converges in distribution to a joint normal distribution with components of the covariance matrix as introduced above, i.e. 
	\begin{equation*}
		\mathbf{U}_{\PFS, \OS} \coloneqq (U_{\PFS}(A_{\PFS}), U_{\OS}(A_{\PFS}), U_{\PFS}(A_{\OS}), U_{\OS}(A_{\OS})) \overset{\mathcal{D}}{\to} \mathcal{N}(0, \boldsymbol{\Sigma}_{\PFS, \OS})
	\end{equation*}
	with
	\begin{equation*}
		\boldsymbol{\Sigma}_{\PFS, \OS} = 
		\begin{pmatrix}
			\sigma^2_{\PFS}(t_{\PFS}) & \sigma_{\PFS, \OS}(t_{\PFS}, t_{\PFS}) & \sigma^2_{\PFS}(t_{\PFS}) & \sigma_{\PFS, \OS}(t_{\PFS}, t_{\OS})\\
			\sigma_{\PFS, \OS}(t_{\PFS}, t_{\PFS}) & \sigma^2_{\OS}(t_{\PFS}) & \sigma_{\PFS, \OS}(t_{\OS}, t_{\PFS}) & \sigma^2_{\OS}(t_{\PFS})\\
			\sigma^2_{\PFS}(t_{\PFS}) & \sigma_{\PFS, \OS}(t_{\OS}, t_{\PFS}) & \sigma^2_{\PFS}(t_{\OS}) & \sigma_{\PFS, \OS}(t_{\OS}, t_{\OS})\\
			\sigma_{\PFS, \OS}(t_{\PFS}, t_{\OS}) & \sigma^2_{\OS}(t_{\PFS}) & \sigma_{\PFS, \OS}(t_{\OS}, t_{\OS}) & \sigma^2_{\OS}(t_{\OS})
		\end{pmatrix}
	\end{equation*}
\end{theorem}
\begin{proof}
	First, note that $\mathbf{U}_{\PFS, \OS}$ is asymptotically equivalent to 
	\begin{equation*}
		\mathbf{u}_{\PFS,\OS} \coloneqq (u_{\PFS}(t_{\PFS}), u_{\OS}(t_{\PFS}), u_{\PFS}(t_{\OS}), u_{\OS}(t_{\OS})),
	\end{equation*}
	i.e. their difference vanishes in probability as $n \to \infty$. To see this, note first, that convergence in probability of a vector reduces to convergence in probability of its components (see Theorem 2.7 (vi) of \cite{vanderVaart:2000}). Let $U$ resp. $u$ denote one component of those vectors and $A$ and $t_0$ denote the random calendar time and its limit, respectively. Then, we have	
	\begin{equation*}
		|U(A) - u(t_0)| \leq |U(A) - U(t_0)| + |U(t_0) - u(t_0)|.
	\end{equation*}
	The standard theory for sequential analysis of log-rank tests (see e.g. \cite{Tsiatis:1981}) yields convergence to 0 in probability for the second summand. Hence, it remains to show convergence in probability of the first summand. Therefor, we check the prerequesites of Lemma \ref{lemma:convergence_random_stopping_time} to obtain convergence of the first summand. Lemma \ref{lemma:calendar_dates} yields convergence of $A$. The required convergence in distribution of $U$ is stated in \cite{Sellke:1983} and (in calendar time) in \cite{Olschewski:1986}. The time transformation is given by $F^{uc}_{D_{\PFS}}$ or $F^{uc}_{D_{\OS}}$, respectively. Their continuity is guaranteed by the standard assumptions (i.e. absolutely continuous distribution and independence of survival time, recruitment date and time to drop-out).\\
	Now, it is shown in \cite{Lin:1991} that $\mathbf{u}_{\PFS, \OS}$ converges in distribution to the normal distribution with the covariance matrix shown above. Finally, we apply Theorem 2.7 (iv) of \cite{vanderVaart:2000} to obtain the same convergence for $\mathbf{U}_{\PFS,\OS}$.
\end{proof}

Now, we address the estimation of $\boldsymbol{\Sigma}_{\PFS, \OS}$. Basically, the proof follows the lines of \cite{Wei:1984}. However, a little more complexity is added, since we consider multi-stage designs and we need statements about the uniform convergence of different processes in a range around the the limiting, fixed calendar dates due to the analysis at the random dates. We obtain these from the theory of empirical processes (see e.g. \cite{Shorack:2009}). For the sake of simplicity, we denote 

\begin{equation*}
	\mu_{E}(t,s) \coloneqq 1 - \frac{y_E^{Z=1}(t,s)}{y_E(t,s)} \quad \text { and } \quad \hat{\mu}_{E}(t,s) \coloneqq 1 - \frac{Y_E^{Z=1}(t,s)}{Y_E(t,s)}.
\end{equation*}
Under the strict null hypothesis of equal distributions of both endpoints in both groups and standard assumptions of equal censoring in both groups, $\mu$ amounts to $1/2$ for any $s < t$.\\
By $\phi$ and $\psi$ we denote the quantity that is obtained when $\mu$ or its respective counterpart $1 - \mu$ is integrated w.r.t. the hazard function, i.e.
\begin{equation*}
	\psi_{E}(t,s) \coloneqq \int_0^s \mu_E(t,u) \Lambda_E(du) \quad \text { and } \quad \hat{\psi}_{E}(t,s) \coloneqq \int_0^s \hat{\mu}_E(t,u) \hat{\Lambda}_E(t,du)
\end{equation*}
and 
\begin{equation*}
	\phi_{E}(t,s) \coloneqq \int_0^s (1 - \mu_E(t,u)) \Lambda_E(du) \quad \text { and } \quad \hat{\phi}_{E}(t,s) \coloneqq \int_0^s (1 - \hat{\mu}_E(t,u)) \hat{\Lambda}_E(t,du),
\end{equation*}
respectively. Under the same assumptions, $\psi(t,s)$ and $\phi(t,s)$ amount to $\Lambda(s)/2$. Additionally, we define
\begin{align*}
	\eta_{E_1, E_2}(t_1, t_2, s) \coloneqq & \int_{0}^{t_1} \int_0^{t_2} \mu_{E_1}(t_1, u) \mathbbm{1}_{v \geq s} \; d\mathbb{P}[X_{E_1}(t_1) \leq u, \Delta_{E_1}(t_1) = 1, X_{E_2}(t_2) \leq u]\\
	= & \mathbb{E}[\mu_{E_1}(t_1, X_{\PFS}(t_1)) \cdot \Delta_{E_1}(t_1) \cdot Y_{E_2}(t_2,s)]
\end{align*}
and its estimator
\begin{equation*}
	\hat{\eta}_{E_1, E_2}(t_1, t_2, s) \coloneqq \frac{1}{n} \sum_{i=1}^n \Delta_{E_1,i}(t_1) \cdot \hat{\mu}_{E_1}(t_1, X_{E_1,i}(t_1))\cdot Y_{E_2,i}(t_2,s)
\end{equation*}
for $E_1, E_2 \in \{\PFS, \OS\}$ and $E_1 \neq E_2$. Note that a uniform bound on $\hat{\mu}$ by some constant $C$ also implies a uniform bound on $\hat{\eta}$ of $C \cdot Y_{E_2}(t_2, s)/n$.
\begin{theorem}\label{theorem:consistency_variance}
	The components of the asymptotic covariance matrix of $\mathbf{U}_{\PFS, \OS}$ can be consistently estimated. The corresponding estimates are given as follows.
	\begin{enumerate}[label=(\roman*)]
		\item\label{item:within_endpoint} For components that refer to the same endpoint, i.e. those of the form $\sigma^2_{E}(t)$ for $E \in \{\PFS, \OS\}$ we can apply standard variance estimates for log-rank statistics. When analysed at the random analysis date $A$ with $A \overset{\mathbb{P}}{\to} t$, this amounts to
		\begin{equation*}
			\hat{\sigma}^2_{E_1}(A) \coloneqq \frac{1}{n} \sum_{i=1}^n \int_0^{A} \frac{Y_{E_1}^{Z=1}(A,s)}{Y_{E_1}(A,s)}  \left( 1 - \frac{Y_{E_1}^{Z=1}(A,s)}{Y_{E_1}(A,s)}   \right) N_{E_1,i}(A,ds).
		\end{equation*}
		\item\label{item:between_endpoints} For covariance estimates for endpoints $E_1 \neq E_2$ that are analysed at random analysis dates $A_1$ and $A_2$ with $(A_1, A_2) \overset{\mathbb{P}}{\to} (t_1, t_2)$, respectively, the covariance $\sigma_{E_2, E_1}(t_1, t_2)$ is estimated by
		\begin{align*}
			&\hat{\sigma}_{\PFS, \OS}(A_1, A_2) \\
			\coloneqq & \frac{1}{n} \sum_{\substack{i=1 \\ Z_i = 1}}^{n} \Bigg( \left( \hat{\mu}_{\PFS}(A_1, X_{\PFS,i}(A_1)) \Delta_{\PFS,i}(A_1) - \hat{\psi}_{\PFS}(A_1, X_{\PFS,i}(A_1)) \right) \cdot \\
			&\qquad \qquad \left( \hat{\mu}_{\OS}(A_2, X_{\OS,i}(A_2)) \Delta_{\OS,i}(A_2) - \hat{\psi}_{\OS}(A_2, X_{\OS,i}(A_2)) \right) \Bigg)\\
			&\quad + \frac{1}{n} \sum_{\substack{i=1 \\ Z_i = 0}}^{n} \Bigg( \left( [1 - \hat{\mu}_{\PFS}(A_1, X_{\PFS,i}(A_1))] \Delta_{\PFS,i}(A_1) - \hat{\phi}_{\PFS}(A_1, X_{\PFS,i}(A_1)) \right) \cdot \\
			&\qquad \qquad \left( [1-\hat{\mu}_{\OS}(A_2, X_{\OS,i}(A_2))] \Delta_{\OS,i}(A_2) - \hat{\phi}_{\OS}(A_2, X_{\OS,i}(A_2)) \right) \Bigg)
		\end{align*}
	\end{enumerate}
\end{theorem}
\begin{proof}
	For both parts of the proof, we want to remind that $A_{E} \overset{\mathbb{P}}{\to} t_{E}$ for $E \in \{\PFS, \OS\}$ and hence also $(A_{\PFS}, A_{\OS}) \overset{\mathbb{P}}{\to} (t_{\PFS}, t_{\OS})$.\\
	\underline{Components of the form \ref{item:within_endpoint}:}\\
	For components that refer to the same endpoint, we can use the well-known variance estimator of log-rank test statistics. We refer to well-known results of sequential analysis of log-rank statistics, as e.g. from \cite{Sellke:1983} to show that this convergence is uniform. Hence, we can pass from $A_{E}$ to $t_{E}$ in the limit of $n \to \infty$.\\
	\underline{Components of the form \ref{item:between_endpoints}:}\\
	This proof follows along similar lines as the one of \cite{Wei:1984}. However, we have to apply a bit more caution as we also have to deal with the random analysis dates $A_E$ when estimating correlations. This makes the use of the following results necessary:
	\begin{enumerate}[label=(\alph*)]
		\item\label{item:unif_conv_1} For any $c > 0$ and $E \in \{\PFS, \OS\}$, on any compact $D^{\star}$ subset of
		\begin{equation*}
			D^{E}_{c} \coloneqq \{(t,s) \colon t \in [t_{E_1}-c, t_{E_1}+c], s < t\}
		\end{equation*}
		it holds
		\begin{equation*}
			\sup_{(t,s) \in D^{\star}} |\hat{\mu}_{E}(t,s) - \hat{\mu}_{E}(t,s)| \overset{\mathbb{P}}{\to} 0
		\end{equation*}
		and
		\begin{equation*}
			\sup_{(t,s) \in D^{\star}} |\hat{\Lambda}_{E}(t,s) - \hat{\Lambda}_{E}(s)| \overset{\mathbb{P}}{\to} 0
		\end{equation*}
		\item\label{item:unif_conv_2} For any $c > 0$, $E_1, E_2 \notin \{\PFS, \OS\}$ with $E_1 \neq E_2$ we have on the set
		\begin{equation*}
			\bar{D}^{E_1, E_2}_{c} \coloneqq \{(t,s) \colon t_1 \in [t_{E_1}-c, t_{E_1}+c], t_2 \in [t_{E_2}-c, t_{E_2}+c], s \leq t_2\}
		\end{equation*}
		the uniform convergence
		\begin{equation*}
			\sup_{(t,s) \in \bar{D}^{E_1, E_2}_{c}} |\hat{\eta}_{E_1, E_2}(t_1, t_2 ,s) - \eta_{E_1, E_2}(t_1, t_2, s)| \overset{\mathbb{P}}{\to} 0.
		\end{equation*}
	\end{enumerate}
	For proof of \ref{item:unif_conv_1}, we refer to Example 2 of \cite{Gu:1991} and Remark 2, Proof of Theorem 2.2 and Lemma A.3 of \cite{Bilias:1997}. In particular, Example 2 of \cite{Gu:1991} also enables the use of weighted log-rank statistics, e.g. those of the Fleming-Harrington class. For the sake of simplicity, we restrict ourselves to the consideration of standard log-rank tests here. The basic idea is the combination of pointwise convergence (as given by standard result in \cite{Andersen:2012}) with tightness of the sequence of measures. Following Prokhorov's Theorem this implies the uniform convergences given above. The proof of \ref{item:unif_conv_2} follows from the convergence result of $\hat{\mu}$, the uniform bound of $\hat{\mu}$ and uniform convergence of (multivariate) empirical measures as presented e.g. in Section 26 of \cite{Shorack:2009}. To see how this results are applied, we decompose as in \cite{Wei:1984}:
	\begin{align*}
		&|\hat{\eta}_{E_1, E_2}(t_1, t_2, s) - \eta_{E_1, E_2}(t_1, t_2, s)|\\
		\leq&\left| \frac{1}{n} \int_{\substack{u \in [0,t_1] \\ v \in [0,t_2]}} \hat{\mu}_{E_1}(t_1, u) \mathbbm{1}_{s \leq v} - \mu_{E_1}(t_1, u) \mathbbm{1}_{s \leq v} \; d\left[\sum_{i=1}^n Y_{E_1,i}(t_1, u) \cdot \Delta_{E_1,i}(t_1) \cdot Y_{E_2,i}(t_2, v)\right] \right|\\
		& + \left| \frac{1}{n} \sum_{i=1}^n \mu_{E_1}(t_1, X_{E_1,i}(t_1)) \cdot \Delta_{E_1,i} \cdot Y_{E_2,i}(t_1,s) - \eta(t_1, t_2, s) \right|
	\end{align*}
	For the first part, uniform convergence of $\hat{\mu}$ on compact subspaces and uniform boundedness can be applied. The second summand is a difference between an expected value and its empirical counterpart which can be uniformly bounded by well-known results of empirical process theory (see Section 26 of \cite{Shorack:2009}).\\
	As in the Appendix of \cite{Wei:1984} we can decompose $\sigma_{\PFS, \OS}(t_{E_1}, t_{E_2})$ into two summands. One concerns observation from group $Z=0$ and the other one those from group $Z=1$. Those can be plugged together to obtain the complete variance. Without loss of generality, we focus on the group $Z=1$. Hence, the following probability statements can all be considered as restricted to $Z=1$. As in \cite{Wei:1984}, this quantity, which we call $\sigma^{Z=1}_{\PFS, \OS}(t_{E_1}, t_{E_2})$, can be decomposed as follows:
	\begin{align}
		\sigma^{Z=1}_{\PFS, \OS}(t_1, t_2)=&\mathbb{E}[\Delta_{\PFS}(t_1) \Delta_{\OS}(t_2) \mu_{\PFS}(t_1, X_{\PFS}(t_1)) \mu_{\OS}(t_2, X_{\OS}(t_2))] \label{eq:limit_1}\\
		&-\mathbb{E}[\Delta_{\PFS}(t_1) \mu_{\PFS}(t_1, X_{\PFS}(t_1)) \psi_{\OS}(t_2, X_{\OS}(t_2))] \label{eq:limit_2}\\
		&-\mathbb{E}[\Delta_{\OS}(t_2) \mu_{\OS}(t_2, X_{\OS}(t_1)) \psi_{\PFS}(t_1, X_{\PFS}(t_1))] \label{eq:limit_3}\\
		&+\mathbb{E}[\psi_{\PFS}(t_1, X_{\PFS}(t_1))\psi_{\OS}(t_2, X_{\OS}(t_2))]. \label{eq:limit_4}
	\end{align}
	Analogously, we can write the corresponding part to estimate this quantity by
	\begin{align}
		\hat{\sigma}^{Z=1}_{\PFS, \OS}(t_1, t_2)=&\frac{1}{n} \sum_{\substack{i=1 \\ Z_i=1}}^n \hat{\mu}_{\PFS}(A_1, X_{\PFS,i}(A_1))\Delta_{\PFS,i}(A_1)\hat{\mu}_{\OS}(A_2, X_{\OS,i}(A_2))\Delta_{\OS,i}(A_2) \label{eq:emp_1}\\
		& - \frac{1}{n} \sum_{\substack{i=1 \\ Z_i=1}}^n \hat{\mu}_{\PFS}(A_1, X_{\PFS,i}(A_1))\Delta_{\PFS,i}(A_1) \hat{\psi}_{\OS}(A_2, X_{\OS,i}(A_2)) \label{eq:emp_2}\\
		& - \frac{1}{n} \sum_{\substack{i=1 \\ Z_i=1}}^n \hat{\mu}_{\OS}(A_2, X_{\OS,i}(A_2))\Delta_{\OS,i}(A_2) \hat{\psi}_{\PFS}(A_1, X_{\PFS,i}(A_1)) \label{eq:emp_3} \\
		& + \frac{1}{n} \sum_{\substack{i=1 \\ Z_i=1}}^n \hat{\psi}_{\PFS}(A_1, X_{\PFS,i}(A_1)) \hat{\psi}_{\OS}(A_2, X_{\OS,i}(A_2)) \label{eq:emp_4}
	\end{align}
	We consider all of the summands separately:\\
	\underline{Convergence of \eqref{eq:emp_1} to \eqref{eq:limit_1}:}\\
	We denote the difference of the two terms by $W$. For an arbitrarily small $\varepsilon > 0$, we have to show $\mathbb{P}[|W > \varepsilon|] \to 0$. In particular, we want to show that $\mathbb{P}[|W > \varepsilon|] < \delta$ for some arbitrarily small $\delta>0$ for $n$ big enough. We can split up
	\begin{align*}
		\mathbb{P}[|W| > \varepsilon]=&\mathbb{P}[|W|>\varepsilon; A_1 \in [t_1 - \gamma, t_1 + \gamma] \text{ and } A_2 \in [t_2 - \gamma, t_2 + \gamma]]\\
		&+\mathbb{P}[|W|>\varepsilon; A_1 \notin [t_1 - \gamma, t_1 + \gamma] \text{ and } A_2 \notin [t_2 - \gamma, t_2 + \gamma]].
	\end{align*}
	The second summand is dominated by $\mathbb{P}[A_1 \notin [t_1 - \gamma, t_1 + \gamma] \text{ and } A_2 \notin [t_2 - \gamma, t_2 + \gamma]]$ and by the convergence of $(A_1, A_2)$ it becomes arbitrarily small for any fixed $\gamma$ as $n \to \infty$. Hence, we can restrict ourselves to considerations conditional on the event in the first summand.\\
	Both \eqref{eq:emp_1} and \eqref{eq:limit_1} can be split up into summands that refer to events that happen earlier and later in calendar time. More explicitly, we can write
	\begin{align*}
		&\mathbb{E}[\Delta_{\PFS}(t_1) \Delta_{\OS}(t_2) \mu_{\PFS}(t_1, X_{\PFS}(t_1)) \mu_{\OS}(t_2, X_{\OS}(t_2))]\\
		=&\mathbb{E}[\Delta_{\PFS}(t_1 - \tau) \Delta_{\OS}(t_2 - \tau) \mu_{\PFS}(t_1, X_{\PFS}(t_1)) \mu_{\OS}(t_2, X_{\OS}(t_2))]\\
		&+\mathbb{E}[\mu_{\PFS}(t_1, X_{\PFS}(t_1))\mu_{\OS}(t_2, X_{\OS}(t_2)); \Delta_{\PFS}(t_1) - \Delta_{\PFS}(t_1 - \tau) = 1 \text{ or } \Delta_{\OS}(t_2) - \Delta_{\OS}(t_2 - \tau) = 1]
	\end{align*}
	We can choose $\tau$ small enough s.t. the second summand is smaller then $\varepsilon/5$. Analogously, this can be done for \eqref{eq:emp_1} by
	\begin{align*}
		&\frac{1}{n} \sum_{\substack{i=1 \\ Z_i=1}}^n \hat{\mu}_{\PFS}(A_1, X_{\PFS,i}(A_1))\Delta_{\PFS,i}(A_1)\hat{\mu}_{\OS}(A_2, X_{\OS,i}(A_2))\Delta_{\OS,i}(A_2)\\
		=&\frac{1}{n} \sum_{\substack{i=1 \\ Z_i=1}}^n \hat{\mu}_{\PFS}(A_1, X_{\PFS,i}(A_1))\Delta_{\PFS,i}(A_1 - \tau)\hat{\mu}_{\OS}(A_2, X_{\OS,i}(A_2))\Delta_{\OS,i}(A_2 - \tau)\\
		&+\frac{1}{n} \sum_{\substack{i=1 \\ Z_i=1}}^n \hat{\mu}_{\PFS}(A_1, X_{\PFS,i}(A_1))\hat{\mu}_{\OS}(A_2, X_{\OS,i}(A_2))\mathbbm{1}_{\Delta_{\PFS,i}(A_1) - \Delta_{\PFS,i}(A_1 - \tau) = 1 \text{ or } \Delta_{\OS,i}(A_2) - \Delta_{\OS,i}(A_2 - \tau) = 1}
	\end{align*}
	The second summand is bounded by the empirical rate of events that happen close to the analysis dates. As in the preceding argument the probability of such an event becomes smaller than $\varepsilon/6$ for $\tau$ small enough. In an area around $t_1$ and $t_2$, we can bound this probability by $\varepsilon/6$ for $\tau$ small enough. Now, as a result of empirical process theory, we have a uniform convergence of the empirical rates to the true probabilities in the said area around $t_1$ and $t_2$. Hence, for large enough $n$ the probability of this second summand to be larger than $\varepsilon/5$ is arbitrarily small.\\
	With these two steps, we restricted ourselves to analysis dates close to the limits and to events, that are bounded away (in calendar time) by $\tau$ from the analysis date. The remaining difference can be decomposed as follows:
	\begin{align*}
		&\frac{1}{n} \sum_{\substack{i=1 \\ Z_i=1}}^n \hat{\mu}_{\PFS}(A_1, X_{\PFS,i}(A_1))\Delta_{\PFS,i}(A_1 - \tau)\hat{\mu}_{\OS}(A_2, X_{\OS,i}(A_2))\Delta_{\OS,i}(A_2 - \tau)\\
		-&\frac{1}{n} \sum_{\substack{i=1 \\ Z_i=1}}^n \mu_{\PFS}(A_1, X_{\PFS,i}(A_1))\Delta_{\PFS,i}(A_1 - \tau)\mu_{\OS}(A_2, X_{\OS,i}(A_2))\Delta_{\OS,i}(A_2 - \tau)\\
		+&\frac{1}{n} \sum_{\substack{i=1 \\ Z_i=1}}^n \mu_{\PFS}(A_1, X_{\PFS,i}(A_1))\Delta_{\PFS,i}(A_1 - \tau)\mu_{\OS}(A_2, X_{\OS,i}(A_2))\Delta_{\OS,i}(A_2 - \tau)\\
		-&\frac{1}{n} \sum_{\substack{i=1 \\ Z_i=1}}^n \mu_{\PFS}(t_1, X_{\PFS,i}(t_1))\Delta_{\PFS,i}(t_1 - \tau)\mu_{\OS}(t_2, X_{\OS,i}(t_2))\Delta_{\OS,i}(t_2 - \tau)\\
		+&\frac{1}{n} \sum_{\substack{i=1 \\ Z_i=1}}^n \mu_{\PFS}(t_1, X_{\PFS,i}(t_1))\Delta_{\PFS,i}(t_1 - \tau)\mu_{\OS}(t_2, X_{\OS,i}(t_2))\Delta_{\OS,i}(t_2 - \tau)\\
		-&\mathbb{E}[\Delta_{\PFS}(t_1) \Delta_{\OS}(t_2) \mu_{\PFS}(t_1, X_{\PFS}(t_1)) \mu_{\OS}(t_2, X_{\OS}(t_2))].
	\end{align*}
	The first and second summand can be bounded by taking the supremum over $A_l \in [t_l - \gamma, t_l + \gamma]$ for both $l \in \{1,2\}$ and then applying the uniform convergence from \ref{item:unif_conv_1} to show convergence to $0$ in probability. The next two terms also converge to $0$ by the Continuous Mapping Theorem. Convergence to zero of the last two summands is given by a simple Law of Large Numbers. Hence, for all of those summands we can choose $n$ large enough s.t. the probability of those summands exceeding $\varepsilon/5$ becomes smaller than $\delta/5$. We can plug all of this together to obtain the desired statement.\\ 
	\underline{Convergence of \eqref{eq:emp_2} to \eqref{eq:limit_2}:}\\
	Summand \eqref{eq:limit_2} can be rewritten by Fubini's theorem as
	\begin{equation*}
		\int_0^{\infty} \mu_{\OS}(t_1, s) \eta_{\PFS, \OS}(t_1, t_2, s) d\Lambda_{\OS}(s). 
	\end{equation*}	
	Analogously, we can reorder \eqref{eq:emp_2} as
	\begin{equation*}
		\int_0^{\infty} \hat{\mu}_{\OS}(A_2, s) \hat{\eta}_{\PFS, \OS}(A_1, A_2, s) \hat{\Lambda}_{\OS}(A_2, ds).
	\end{equation*}
	We can split both terms up by
	\begin{align*}
		&\int_0^{t_2 - \tau} \mu_{\OS}(t_1, s) \eta_{\PFS, \OS}(t_1, t_2, s) d\Lambda_{\OS}(s)\\
		+&\int_{t_2 - \tau}^{\infty} \mu_{\OS}(t_1, s) \eta_{\PFS, \OS}(t_1, t_2, s) d\Lambda_{\OS}(s).  
	\end{align*}
	and 
	\begin{align*}
		&\int_0^{A_2 - \tau} \hat{\mu}_{\OS}(A_2, s) \hat{\eta}_{\PFS, \OS}(A_1, A_2, s) \hat{\Lambda}_{\OS}(A_2, ds)\\
		+&\int_{A_2 - \tau}^{\infty} \hat{\mu}_{\OS}(A_2, s) \hat{\eta}_{\PFS, \OS}(A_1, A_2, s) \hat{\Lambda}_{\OS}(A_2, ds).  
	\end{align*}
	As in the previous part of the proof, the respective second summands of those two terms become arbitrarily small for some $\tau$ small enough if $n$ is big enough. For the second term, we require the fact that $\hat{\eta}_{\PFS, \OS}(A_1, A_2, s)$ is bounded by $Y_{\OS}(A_2,s)/n$.\\
	As above, we can also restrict ourselves to considerations of random calendar dates $A_1$ and $A_2$ that are close to their limits. We can rewrite the remaining difference as
	\begin{align*}
		&\int_{0}^{A_2 - \tau} \hat{\mu}_{\OS}(A_2, s) \hat{\eta}_{\PFS, \OS}(A_1, A_2, s) \; d(\Lambda_{\OS}(s) - \hat{\Lambda}_{\OS}(A_2, s))\\
		+ & \int_{0}^{A_2 - \tau} \left( \hat{\mu}_{\OS}(A_2, s) \hat{\eta}_{\PFS, \OS}(A_1, A_2, s) - \mu_{\OS}(A_1, s) \eta_{\PFS, \OS}(A_1, A_2, s) \right) d\Lambda_{\OS}(s)\\
		+ & \int_{0}^{A_2 - \tau} \mu_{\OS}(A_1, s) \eta_{\PFS, \OS}(A_1, A_2, s)\; d\Lambda_{\OS}(s) - \int_{0}^{t_2 - \tau} \mu_{\OS}(t_1, s) \eta_{\PFS, \OS}(t_1, t_2, s)\; d\Lambda_{\OS}(s).
	\end{align*}
	The first summand converges to zero by the uniform bounds on $\hat{\mu}$ and $\hat{\eta}$ and the uniform convergence of the Nelson-Aalen estimate stated in \ref{item:unif_conv_1}. The second summand then vanishes in the limit because of the uniform convergence stated in \ref{item:unif_conv_2}. The convergence of the last summand is provided by the Continuous Mapping Theorem.\\
	\underline{Convergence of \eqref{eq:emp_3} to \eqref{eq:limit_3}:}\\
	This convergence can be proven in the same way as the previous one with roles swapped between $\PFS$ and $\OS$.\\
	\underline{Convergence of \eqref{eq:emp_4} to \eqref{eq:limit_4}:}\\
	The arguments used so far can be repeated for this part of the proof. Other necessary adjustments have to be made analogously to the proof in \cite{Wei:1984}.
\end{proof}

\clearpage

\section{Simulation results}

\subsection{Compliance with the FWER}

As a supplement of Figure \ref{figure:fwer_comparisons} of the main manuscript, we provide the FWERs for the remaining testing procedures that are given in Table \ref{table:testing_procedures} but not shown in Figure \ref{figure:fwer_comparisons}. As already mentioned in the main manuscript, the results are very similar to those of the other procedures and there is no obvious evidence that any of these procedures systematically fail to comply with the FWER.

\begin{figure}[ht]
	\centering
	\includegraphics[width=\textwidth]{"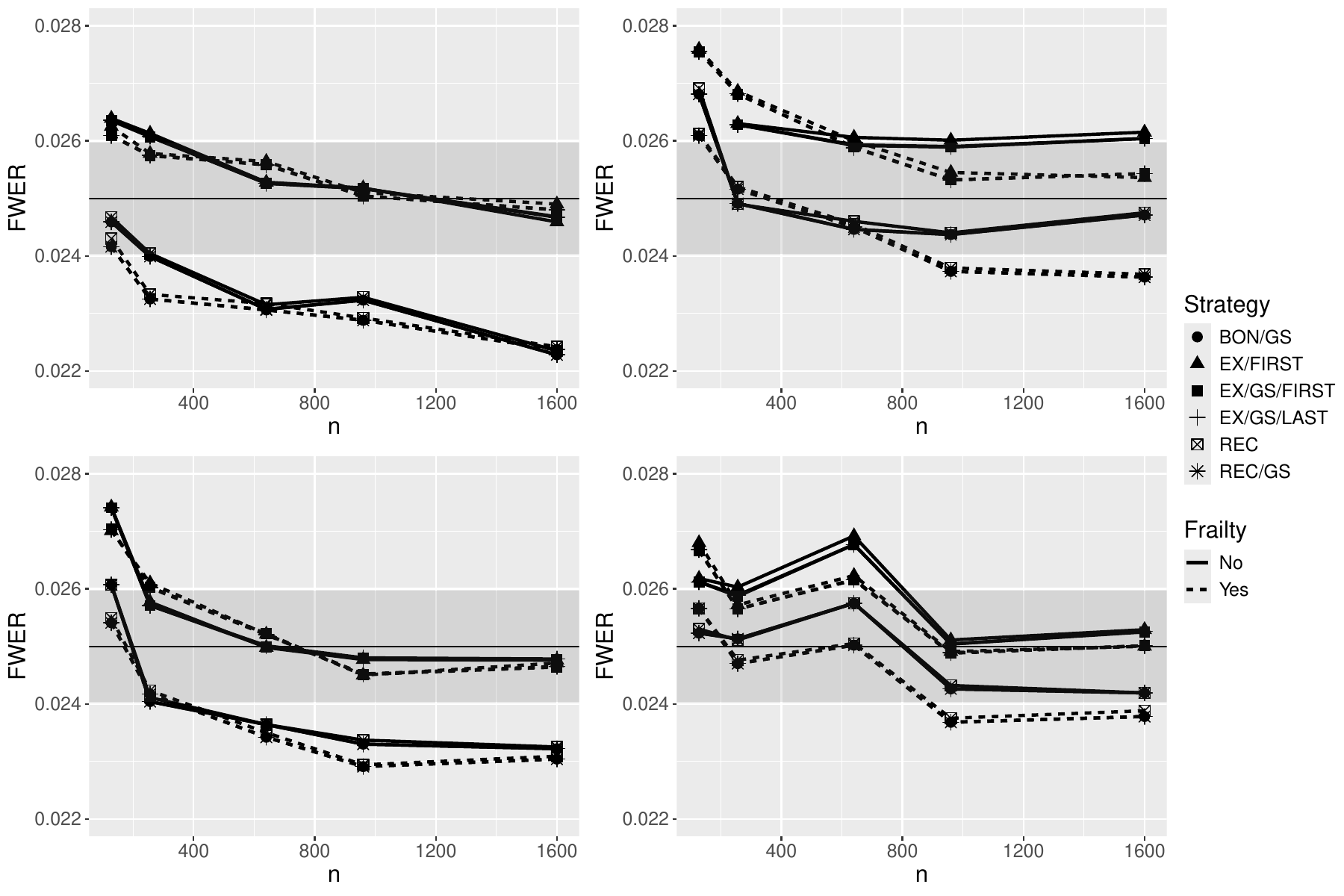"}
	\caption{FWERs of the four testing procedures BON/GS, EX/FIRST, EX/GS/FIRST and EX/GS/LAST with and without consideration of frailties in all four scenarios. The shaded area characterises the Monte Carlo sampling error interval. The subfigures are arranged as follows: Scenario 1, top left; Scenario 2, top right; Scenario 3, bottom left; Scenario 4, bottom right.}
\end{figure}%

\clearpage

\subsection{Power}

Analogously to Figure \ref{figure:power_comparisons} of the main manuscript, we provide power comparisons between some of the testing procedures presented in Table \ref{table:testing_procedures} of the main manuscript for the remaining scenarios (i.e. scenarios 2, 3 and 4 of Table \ref{table:results_scenario1_w=1}) with varying weighting parameter $w$, here. As in the main manuscript, we do not show the power to reject $H_{0,\PFS}$ (on the left hand sides of the following figures) for the procedures EX/GS/FIRST and EX/GS/LAST because their results are very similar to those of EX/FIRST and EX/LAST, respectively. Analogously, we do not show the disjunctive power (on the right hand side of the following figures) for the procedures EX/FIRST and EX/GS/FIRST because their results are very similar to those of EX/LAST and EX/GS/LAST, respectively.

\begin{figure}[ht]
    {\large\bfseries Scenario 2}
    \vspace{6pt}
    
	\centering
	\includegraphics[width=\textwidth]{"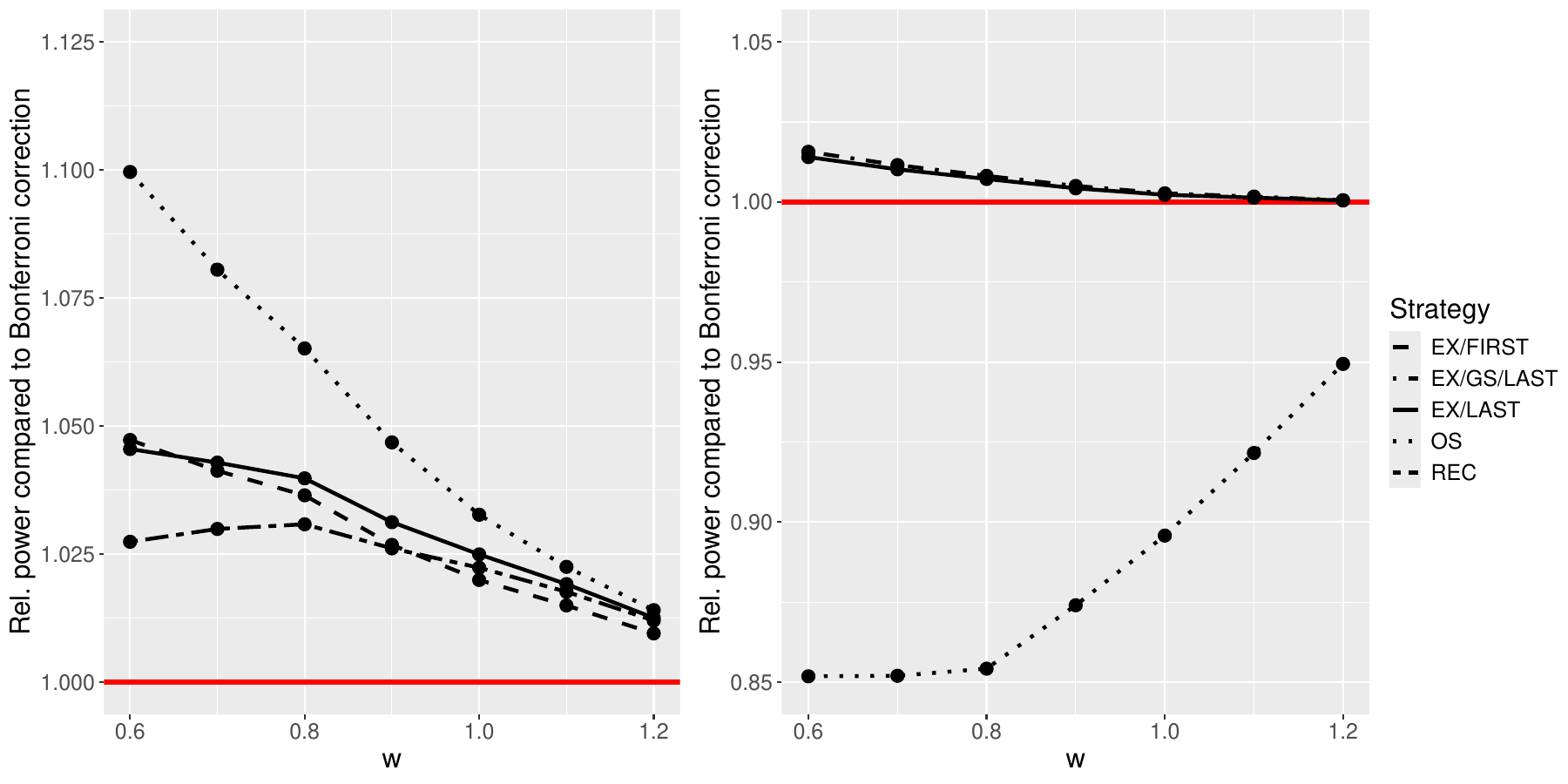"}
	\caption{On the left: Relative differences in power to reject $H_{0,\OS}$ of improved testing procedures and the OS procedure compared to the BON procedure. On the right: Relative differences in disjunctive power of improved testing procedures and the OS procedure compared to the BON procedure.\\
	Please note that not all procedures are shown to ensure the clarity of the plots.}
\end{figure}%

\begin{figure}
    {\large\bfseries Scenario 3}
    \vspace{6pt}
    
	\centering
	\includegraphics[width=\textwidth]{"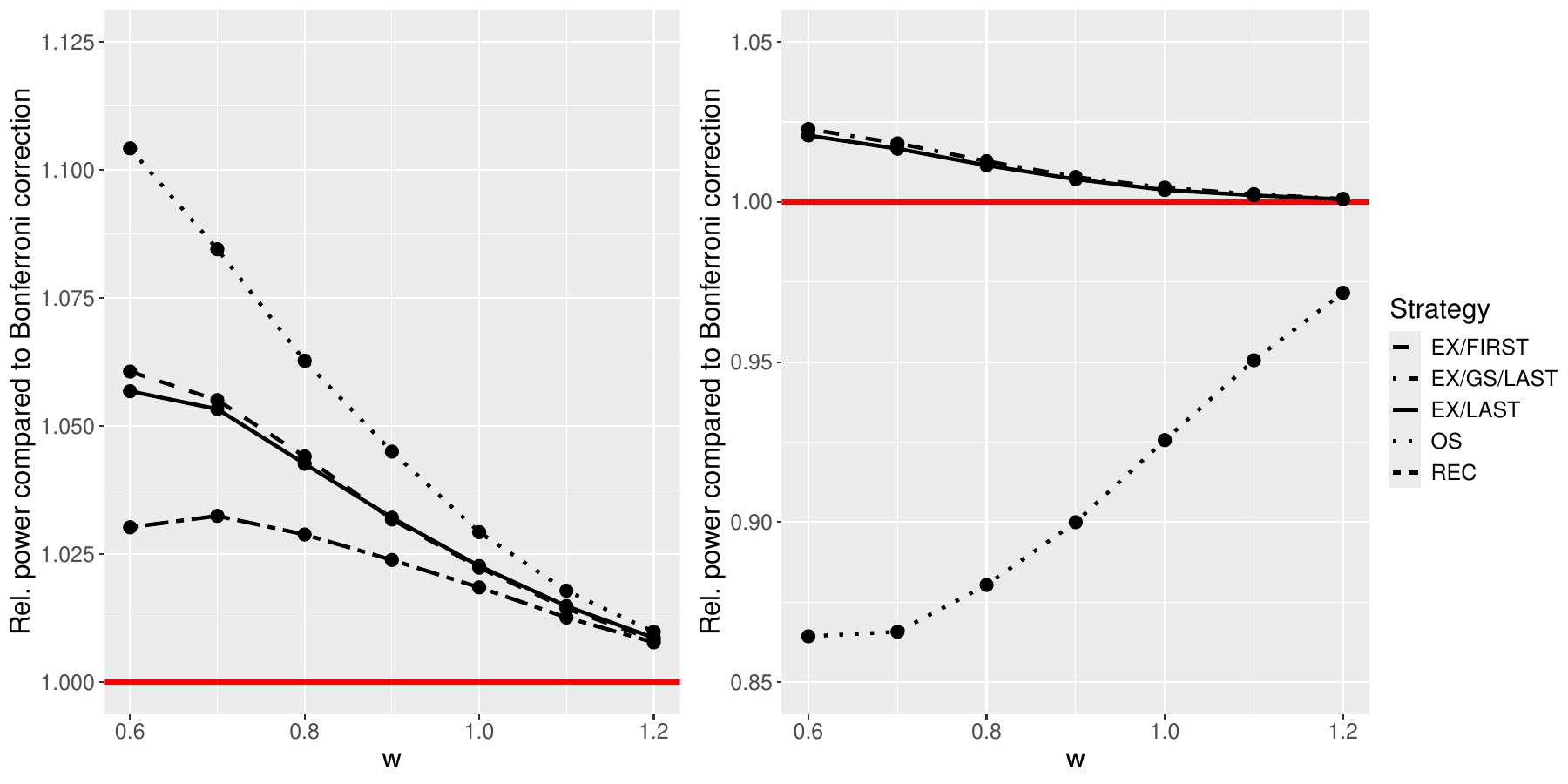"}
	\caption{On the left: Relative differences in power to reject $H_{0,\OS}$ of improved testing procedures and the OS procedure compared to the BON procedure. On the right: Relative differences in disjunctive power of improved testing procedures and the OS procedure compared to the BON procedure.\\
	Please note that not all procedures are shown to ensure the clarity of the plots.}
\end{figure}%

\clearpage

\begin{figure}
    {\large\bfseries Scenario 4}
    \vspace{6pt}
    
	\centering
	\includegraphics[width=\textwidth]{"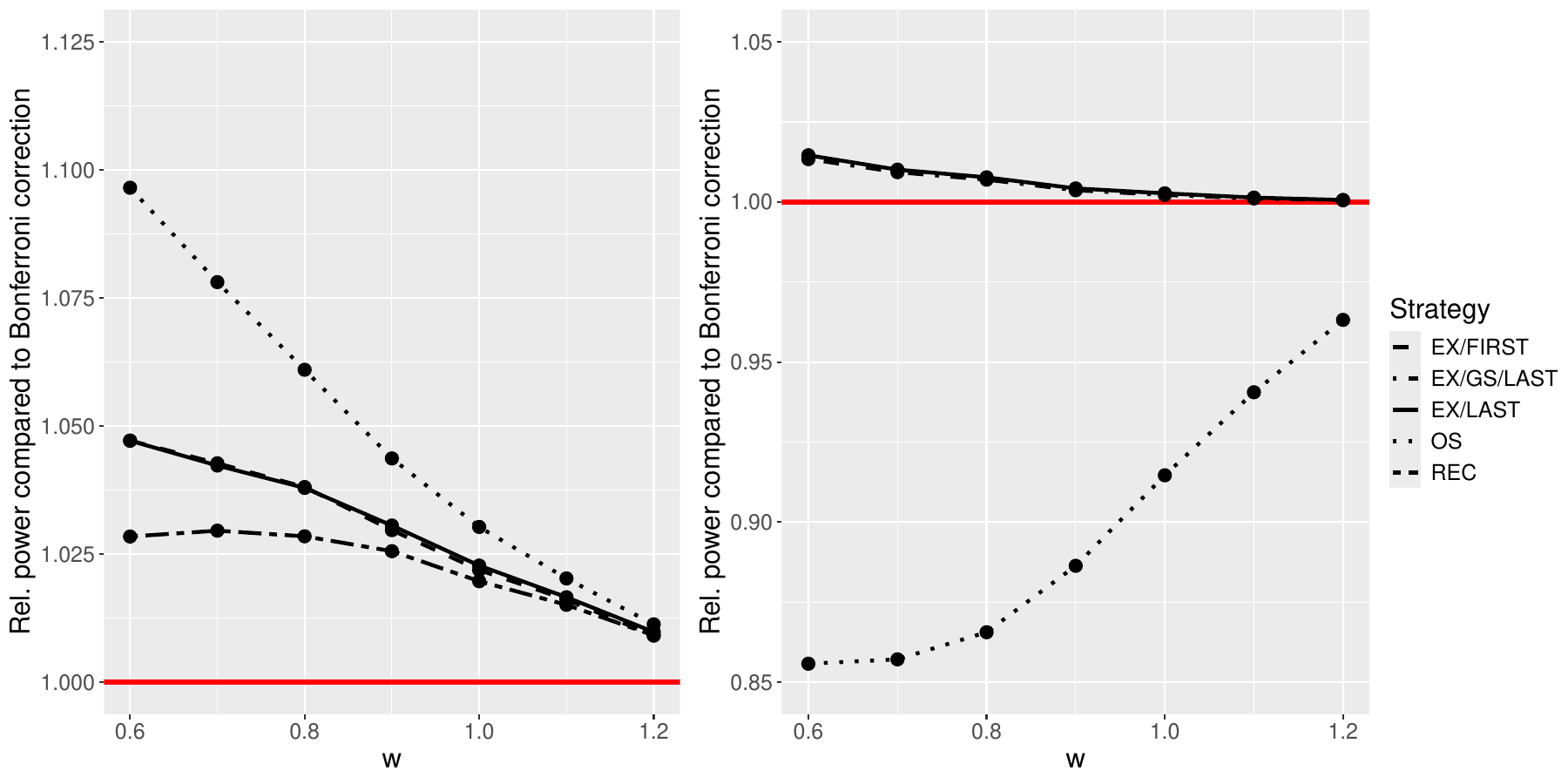"}
	\caption{On the left: Relative differences in power to reject $H_{0,\OS}$ of improved testing procedures and the OS procedure compared to the BON procedure. On the right: Relative differences in disjunctive power of improved testing procedures and the OS procedure compared to the BON procedure.\\
	Please note that not all procedures are shown to ensure the clarity of the plots.}
\end{figure}%

\putbib
\end{bibunit}

\end{document}